\documentclass[12pt,reqno]{amsart}
\usepackage{amsfonts}
\usepackage{pdfsync}
\usepackage{amssymb, amsmath}
\usepackage{graphicx, graphics, color, epsfig, setspace, epstopdf}
\usepackage[section]{placeins}
\def\sqr#1#2{{\vcenter{\vbox{\hrule height.#2pt
        \hbox{\vrule width.#2pt height#1pt \kern#2pt
        \vrule width.#2pt}
        \hrule height.#2pt}}}}

\newcommand{\nc}{\newcommand}
\nc{\parent}[1]{$[\![#1]\!]$}
\newtheorem{theorem}{Theorem}[section]
\newtheorem{lemma}{Lemma}[section]
\newtheorem{example}{Example}[section]
\newtheorem{corollary}{Corollary}[section]
\newtheorem{proposition}{Proposition}[section]
\newtheorem{remark}{Remark}
\newtheorem{definition}{Definition}[section]
\newtheorem{assumption}{Assumption}[section]

\newenvironment{pf-main}{{\bf \sc Proof of Theorem \ref{mainresult}.}\hspace{3mm}}{\qed}
\nc{\cadlag}{c\`{a}dl\`{a}g } \nc{\ba}{\begin{array}}
\nc{\ea}{\end{array}} \nc{\be}{\begin{equation}}
\nc{\ee}{\end{equation}} \nc{\bea}{\begin{eqnarray}}
\nc{\eea}{\end{eqnarray}} \nc{\bean}{\begin{eqnarray*}}
\nc{\eean}{\end{eqnarray*}} \nc{\bu}{\bullet} \nc{\nn}{\nonumber}
\nc{\cA}{{\mathcal A}} \nc{\cB}{{\mathcal B}} \nc{\cC}{{\mathcal C}} \nc{\bfE}{\mathbf{E}}
\nc{\cD}{{\mathcal D}} \nc{\bbD}{\mathbb{D}}\nc{\bbH}{\mathbb{H}}
\nc{\bbF}{\mathbb{F}}\nc{\bbG}{\mathbb{G}}\nc{\cG}{{\mathcal G}} \nc{\cF}{{\mathcal F}}
\nc{\cS}{{\mathcal S}} \nc{\cU}{{\mathcal U}} \nc{\cH}{{\mathcal H}}
\nc{\cK}{{\mathcal K}} \nc{\cL}{{\mathcal L}} \nc{\cM}{{\mathcal M}}
\nc{\cO}{{\mathcal O}} \nc{\cP}{{\mathcal P}} \nc{\cX}{{\mathcal X}}\nc{\bbE}{\mathbb{E}}
\nc{\bbEQ}{\mathbb{E}_{\mathbb{Q}}} \nc{\eps}{\varepsilon}
\nc{\bbEP}{\mathbb{E}_{\mathbb{P}}}\nc{\bbL}{\mathbb{L}}
\nc{\bbP}{\mathbb{P}} \nc{\bbQ}{\mathbb{Q}} \nc{\del}{\partial}
\nc{\Om}{\Omega} \nc{\om}{\omega} \nc{\bbR}{\mathbb{R}}
\nc{\bbC}{\mathbb{C}} \nc{\bfr}{\begin{flushright}}
\nc{\efr}{\end{flushright}} \nc{\dXt}{\Delta X_{t}} \nc{\dXs}{\Delta
X_{s}} \nc{\bs}{\blacksquare} \nc{\dX}{\Delta X} \nc{\dY}{\Delta Y}
\nc{\dnkx}{\left(X(T^{n}_{k})-X(T^{n}_{k-1})\right)}
\nc{\esssup}{\mathrm{ess}\mbox{ }\mathrm{sup}}
\nc{\essinf}{\mathrm{ess}\mbox{ } \mathrm{inf}}
\nc{\dhats}{\widehat{\delta_s}} \nc{\tX}{\tilde{X}}
\nc{\tZ}{\tilde{Z}}
\nc{\what}{\widehat}
 \nc{\half}{\frac{1}{2}}

\def\rar{\rightarrow} \nc{\uar}{\uparrow}

\nc{\sgn}{\mbox{sgn}}
\nc{\chf}{\mbox{$\mathbf1$}} \nc{\eid}{\stackrel{d}{=}}

\begin{document}

\title[Informed trading and the limit order book]{Informed trading, limit order book and implementation shortfall: equilibrium and asymptotics}
\author{Umut \c{C}et\.in}
\address{Department of Statistics, London School of Economics and Political Science, 10 Houghton st, London, WC2A 2AE, UK}
\email{u.cetin@lse.ac.uk}
\author{Henri Waelbroeck}
\address{Baruch College, CUNY, NY, USA and AlgoCortex LLC}
\email{H.Waelbroeck@algocortex.com}
\date{\today}
\begin{abstract}
We propose a static equilibrium model for limit order book where $N\geq 1$ profit-maximizing investors receive an information signal regarding the liquidation value of the asset and execute via a competitive dealer with random initial inventory. While the dealer's initial position plays a role similar to noise traders in Kyle \cite{Kyle1985}, he trades against a competitive limit order book populated by liquidity suppliers as in Glosten \cite{Glosten94}.  We show that an equilibrium exists for bounded signal distributions, obtain closed form solutions for Bernoulli-type signals and propose a straightforward iterative algorithm to compute the equilibrium order book for the general case.  We obtain the exact analytic asymptotics for the market impact of large trades and show that the functional form depends on the tail distribution of the private signal of the insiders. In particular, the impact follows a power law if the signal has fat tails while the law is logarithmic in case of lighter tails. Moreover, the tail distribution of the trade volume in equilibrium obeys a power law in our model.  We find that  the liquidity suppliers charge a minimum bid-ask spread that is independent of the amount of `noise' trading but increasing in the degree of informational advantage of insiders in equilibrium. The model also predicts that the order book flattens as the amount of noise trading increases converging to a model with proportional transactions costs. In case of a monopolistic insider we show that the {\em last slice} traded against the limit order book is priced at the liquidation value of the asset. However, competition among the insiders leads to aggressive trading causing the aggregate profit to vanish in the limiting case $N\rar \infty$.   The numerical results also show that the spread increases with the number of insiders keeping the other parameters fixed.  Finally, an equilibrium may not exist if the liquidation value is unbounded. We conjecture that existence of equilibrium requires a sufficient amount of competition among insiders if the signal distribution exhibit fat tails. 
\end{abstract}
\maketitle

\section{Introduction}
Kyle \cite{Kyle1985} studies in a simple but remarkably powerful framework a single risk neutral informed trader and a number of non-strategic uninformed liquidity traders submitting orders to a market maker, who aggregates all the orders and clear the market at a single price. Consequently, Kyle's batch trading model does not produce a bid-ask spread. However, his model allows for an explicit characterisation of equilibrium parameters - including the optimal strategy of the informed trader as well as the equilibrium pricing rule. This in turn allows us to analyse how the private information is disseminated to the market and gets incorporated into the prices over time.  In particular, {\em Kyle's lambda} yields an explicit measurement of market's liquidity and price impact of trades. 

While the role of the  market makers and the price setting mechanism of Kyle's batch trading model were in line with the practices of specialists and floor traders of the main exchanges in the the 80s, the role of the designated market makers has diminished in recent years. Nowadays most of the equity and derivative exchanges have moved to the electronic limit order book format.  What distinguishes limit order markets from markets with a uniform market-clearing price as in Kyle \cite{Kyle1985} is that each limit order is executed at its respective limit price leading to {\em discriminatory pricing}. 

The case of liquidity suppliers moving first and submitting limit orders to be later hit by potentially informed market order was first studied by Rock \cite{rock} and Glosten \cite{Glosten94}. In equilibrium even infinitesimal trades have significant impact, which results in a positive bid-ask spread unlike the uniform price auction models. Following the approach initiated by Rock and Glosten,  Seppi \cite{Seppi97} studies the liquidity provision by a specialist competing against a competitive limit order book. His model also allows a discussion on market design issues such as the effect of `tick size.'  Based on the model assumptions of \cite{Seppi97} Parlour and Seppi \cite{ParSep03} build a model of competition between exchanges. In another static model Foucault and Menkveld \cite{FM08} consider an equilibrium among limit order traders and a broker than can be one of two types and use this model to study the effects of market fragmentation in the context of rivalry between Euronext and the London Stock Exchange in the Dutch stock market.

While useful for policy purposes, the above static models have also some undesirable features.    Probably the most unrealistic feature of these models is that, quoting Parlour and Seppi \cite{ParSepSurvey}, ``investors either have an inelastic motive to trade, and are willing to pay for immediate execution via market orders, or they are entirely disinterested liquidity providers with no reason to trade other than to be compensated for supplying liquidity via limit orders." In particular, unlike the Kyle's model, their focus is more concentrated on the supply side of the liquidity and the investors' trading strategies are not completely endogenous. As a result, the shape of the limit order book, which in general can only be obtained numerically, that arises in these models does not truly capture the impact of trades by a rational investor with price elastic motives. 

In this paper we present a static microstructure model for the limit order book, where the adverse selection occurs due to the existence of  informed traders (henceforth called {\em insider}s).  Following Kyle \cite{Kyle1985}, the  insiders know the liquidation value $V$ of the asset in advance and place a market order to maximise their expected profit. They submit their order to a competitive dealer who already has a position that is independent from the liquidation value.  The dealer trades the aggregate amount against a competitive limit order book populated by liquidity suppliers that arrived to the market before the dealer and the insiders.

Our assumption that the limit order book is competitive requires a justification in view of the results of Biais et al. \cite{BMR00} and Back and Baruch \cite{BBLOB}. First of all, the competitive offer curve should be viewed as a limiting book when the number of limit order traders increases to infinity. However, as shown in an example in \cite{BBLOB}  the competitive book is not always obtained as a limit of Nash equilibiria among finitely many liquidity providers if the level of adverse selection faced by the limit order traders is not high enough. In our model there is always sufficient adverse selection since we assume that the dealer's own demand for the asset is normally distributed.

Even though we model the interactions in a single period model, our framework also has the flavour of a continuous-time Kyle model considered by Back \cite{Back}. Indeed, in case of a monopolistic insider, we show that conditional on $V=v$ the last infinitesimal slice of the aggregate order traded by the dealer against the order book is priced at $v$ on average, reminiscing the convergence of the equilibrium price to the liquidation value by the end of the trading horizon in the continuous-time Kyle model. Thus,  all private information gets incorporated into the order book once the last slice has traded. This is  in contrast with  the corresponding Kyle model in one period\footnote{In the corresponding one-period Kyle model, where $V$ is normally distributed with mean $0$, the average price for the aggregate order conditional on $V=v$ equals $\frac{v}{2}$.}, which is not surprising given the uniform pricing in the Kyle model.

The above `convergence' result  motivates the analysis of the case of multiple insiders.  In the same vein as Holden and Subrahmanyam \cite{HolSub92} we also model and solve the imperfect competition among $N \geq 2$ insiders who observe the liquidation value $V$ perfectly. As in \cite{HolSub92} this competition leads to more aggressive trading. We find that (conditional on $V=v$) the average price for the last slice traded by the dealer exceeds $v$ if the insiders are buying in equilibrium, which can be attributed to more aggressive buying due to the competition. An analogous observation holds when the optimal strategy in equilibrium is to sell. We obtain an explicit formula for the aggregate profit of the insiders that shows that the aggregate profit converges to $0$ as $N \rar \infty$.

 We characterise the equilibrium strategy of insiders and the corresponding equilibrium order book as a solution of an integral equation, which is equivalently given by the fixed point of an integral operator.  Although this equation admits an analytic solution only in very restrictive settings, its numerical implementation is straightforward. We establish the existence of equilibrium when the fundamental value is bounded and obtain numerical solutions for a number of unbounded distributions commonly found in the literature and used in practice.

Despite the fact that the exact form of the order book can only be solved numerically, we are nevertheless able to obtain the exact analytic asymptotics when the order size is large. Moreover, due to the scaling property of the normal distribution, these asymptotics still provide a good approximation for small orders if the variance of the dealer's demand is sufficiently small. The error in such an approximation appears to be still low  even in the case of higher variance as  confirmed visually by our numerical studies.

Due to the discriminatory pricing in the order book, our model produces a non-zero bid-ask spread in equilibrium. We find that the bid-ask spread imposed by the limit order traders is the same regardless of the variance of the initial inventory of the dealer.  The shape of the order book, however, is not invariant to the changes in this variance; on the contrary, the order book flattens as  it increases. That is, the limit order traders do not need to extract significantly higher rents for larger trades if the information content of the order is very small. In the limit  the equilibrium converges to one where the transaction cost is proportional to the trade size. Moreover,  similar to the model considered by Foucault \cite{Foucault99}, our model also predicts that a larger price volatility leads to a bigger bid-ask spread - a phenomenon that is empirically documented by Ranaldo \cite{Ranaldo04}. Indeed, we show that the spread gets wider as the variance of the private signal of the insiders increases, which corresponds to the case of a higher informational advantage.

That we can compute the price impact of large trades has profound implications for practitioners. Portfolio managers endeavour to find mispriced assets in order to beat their benchmarks. When they decide to change their holdings to take advantage of a mispricing, they create orders in an order management system; the trading desk routes these orders to broker-dealers for execution. When they do so, on average, the price to buy  is greater than at the time of the decision, and vice-versa, the selling price is lower, because the trade's effect on supply and demand causes market impact. The {\em implementation shortfall} is thereby the value lost by not being able to execute at the decision price \cite{peroldIS}. Understanding how this implementation shortfall scales with trade size is essential to optimize the investment management process. A market impact model is needed in order to maximize net trading profits. It is also important to optimize position sizes, limit liquidity risk and estimate a portfolio's capacity. 

Empirical data provides some clues as to the shape and scale of market impact. In particular, it has been known for some time that impact is a concave function of trade size (see Torre \cite{Torre97}). But biases in the data, a low signal-to-noise ratio and other issues prevent practitioners from determining the functional form of market impact precisely, particularly for large trades where data is sparse. Models used by practitioners include square root and logarithm (e.g. Torre \cite{Torre97}, Potters and Bouchaud \cite{PBimpact},  Almgren et al. \cite{Almgren2005}, Bershova and Rakhlin \cite{Bershova2013}, and Zarinelli et al. \cite{Zarinelli15}). When calibrated to data, these models yield similar results for small to moderate trade sizes, but they yield very different predictions for the impact of for very large trades. This is unfortunate given that the largest trades typically dominate the aggregate implementation shortfall for most portfolios. The absence of a consensus on the shape of market impact for very large trades motivates the search for a theoretical framework that would capture the essential features of trading and yet be sufficiently parsimonious to be testable. 

We show that the market impact follows a power law or a logarithmic law depending on the distribution of the liquidation value. The price impact, and equivalently the implementation shortfall, has a power law if the liquidation value has fat tails while lighter tails lead to a logarithmic behaviour. 

Our framework also allows us to compute the tails of the probability distribution associated with aggregate volume. For a large class of distributions that can be used to model the liquidation value we show that the tail probability distribution for the trade volume  obeys a power law. The power-law behavior of order sizes has been documented previously in various contexts: Gopikrishnan et al. \cite{powerlawVol} showed that the market volume in a time interval $\Delta t$ has a power-law tail with exponent 1.7;  Lillo et al. consider off-book trades \cite{Lillo2005} and find an exponent ranging from 1.59 to 1.74. Vaglica et al. reconstructed metaorders in Spanish stock exchange using data with brokerage codes and found that the metaorder transaction size is distributed has a power law tail with exponent 1.7 \cite{Vaglica2008}. And in a study performed directly on institutional trade data from Alliance Bernstein, Bershova and Rakhlin found a tail exponent 1.56 for metaorder sizes \cite{Bershova2013}.

Our proof of the existence of equilibrium assumes that the liquidation value of the asset is bounded.   However, our numerical experiments suggests that equilibrium exists for a large class of unbounded signals. Moreover, the asymptotics of these solutions agree with the analytical forms that our theory predicts.  However, unlike the bounded case, an equilibrium may not exist for all unbounded distributions. Given the premise of numerical results and our formal calculations we conjecture that existence of equilibrium requires a sufficient amount of competition among insiders if the signal distribution exhibit fat tails. For instance, when the private signal is given by a Student's t-distribution our numerical iterations diverge in case of a monopolistic insider.

We also consider an alternative model where instead of sending the order to a dealer with an existing inventory, the informed investors send the order to an institutional trading desk that also receives orders from noise traders. The aggregate is liquidated against a limit order book; however, in the case of the institutional trading desk the insiders and noise traders all receive allocations at the same average price. While we were not able to find an existence proof for this model, numerical solutions are similar to  those of the dealer inventory model.

Structure of the paper is as follows: We present the dealer inventory model in the next section. In Section 3, we show the existence of an equilibrium and characterize some of its properties. Section 4 considers the asymptotic behaviour of market impact. Numerical solutions are provided in Section 5 for the dealer inventory model. Section 6 explores the trading desk model numerically and compares results with the dealer inventory model.

\section{The market structure and equilibrium}

Our model is built upon Glosten \cite{Glosten94} and the trading takes place in one-period: There are three class of investors that are all risk-neutral: 1) competitive liquidity suppliers who post limit orders, 2) a competitive  dealer who clears the market, and 3) $N\geq 1$ insider(s), who know(s)  the liquidation value $V$ of the asset.  All random variables in this section are assumed to be define on a complete probability space $(\Om, \cF, P)$ and $E$ is the expectation operator associated to $P$.

Liquidity suppliers move first and place limit orders that give rise to an order book.   That is, if the {\em limit order book} is defined by some function $h:\bbR \to \bbR$, the market order moving up (or down) the book faces a cost of $h(y)dy$ as soon as hitting the limit order at level $y$.  The dealer already has $Z \sim N(0,\sigma^2)$ number of shares of the asset, which is assumed to be independent of $V$. The dealer's initial inventory in the asset is unknown by other market participants. The insider chooses a  trade size $X$ to maximise her expected profit conditional on her private information. 

Let us first consider the case $N=1$. We assume that the cost of a market order of $X$ units by the insider is
\be \label{e:insidercost}
\int_0^X h(Z+y)dy.
\ee

The above cost can be justified as a result of Bertrand competition among dealers whose initial mean-zero inventories are normally distributed with identical variance $\sigma^2$. Indeed, first observe that when the insider trades $X$ units via the dealer, the dealer's inventory changes from $Z$ to $Z+X$. If the cost to the insider is given by (\ref{e:insidercost}), the total  profit of the dealer after trading via the liquidity suppliers is given by
\[
\int_0^X h(Z+y)dy-\int_0^{X+Z} h(y)dy=\int_Z^{X+Z} h(y)dy-\int_0^{X+Z} h(y)dy=-\int_0^Z h(y)dy,
\]
which is the cost of liquidating his initial inventory directly via the limit order book. Therefore, no dealer will be willing to charge less than (\ref{e:insidercost}).  The insider knows the distribution of $Z$ but do not know the inventories of individual dealers. Dealers first announce that they will price according to \ref{e:insidercost} plus a premium (the dealer's profit). Insider then chooses a dealer based on this information, before knowing $Z$. Bertrand competition will then drive the dealer's premium to zero. That is, an individual dealer knows that if he charges higher than what is proposed by (\ref{e:insidercost}), he will  be undercut by another dealer. Thus, a Bertrand competition will lead to (\ref{e:insidercost}) for the cost of the trades of the insider. 

\begin{remark}
	Note that the quote given by the competitive dealer that leads to (\ref{e:insidercost})  does not violate the {\em limit order protection rule} that is mandated by the SEC in the US. To see this suppose the insider wants to buy $X>0$ many units. She is charged $p(X):=\frac{1}{X}\int_Z^{X+Z} h(y)dy $ on average for this transaction. If $p(X)$ is bigger than the best ask, i.e. $h(0+)$,  the price priority implies that the limit sell orders at price $p(X)$ and less must be filled first. Thus, the dealer will first buy $h^{-1}(p(X))$ many shares at a cost of $\int_0^{h^{-1}(p(X))}h(y)dy$ and the remaining $X+Z-h^{-1}(p(X))$ shares in his inventory at $\int_{h^{-1}(p(X))}^{X+Z} h(y)dy$.  This makes his cumulative cost $\int_0^{X+Z}h(y)dy$ and final profit $-\int_0^Z h(y)dy$ coinciding with above calculations.
\end{remark}
Thus, in view of (\ref{e:insidercost})  the expected profit of the insider from a market order of size $x$ is given by
\[
E^v\left[Vx-\int_0^x h(Z+y)dy\right],
\]
where $E^v$ is the expectation operator for the insider with the private information $V=v$. 
Since $h$ is assumed nondecreasing, the first order condition characterises the unique  $X^*$ achieving the maximum expected profit via $V=F(X^*)$, where
\be \label{e:foc1}
F(x):=\int_{-\infty}^{\infty}h(x+z)q(\sigma,z)dz
\ee
and $q(\sigma,\cdot)$ is the probability density function of a mean-zero Gaussian random variable with variance $\sigma^2$.
Note that since $h^*$ is non-decreasing and not constant, $F$ is strictly increasing and one-to-one. Thus, $X^*=F^{-1}(V)$.

We shall also consider the case of multiple insiders trading via the same competitive dealer and knowing the value of $V$. Assuming that the dealer charges each insider an amount proportional to their order size, the expected profit of an individual insider placing an order of size $x$ is given by
\[
E^v\left[Vx-\frac{x}{U+x}\int_0^{U+x} h(Z+y)dy\right],
\]
where $U$ denotes the aggregate demand of the other insiders. The number of insiders will be denoted by $N$. 

The first order condition associated with the above optimisation problem of an individual insider is again given by
\[
V=E^v\left [\frac{x}{U +x}h(Z+U+x) +\frac{U}{((U+x)^2}\int_0^{U +x}h(Z+u)du\right].
\]
As every insider has symmetric information and is risk-neutral, the equilibrium demand $x^*$ for each insider must be the same and satisfy
\[
V=E^v\left [\frac{h(Z+Nx^*)}{N} +\frac{N-1}{N^2x^*}\int_0^{Nx^*}h(Z+u)du\right].
\]
Denoting the total informed demand by $X^*$,  the above can be rewritten as $V=F(X^*)$, where
\be \label{e:foc}
F(x):=E^v\left [\frac{h(Z+x)}{N} +\frac{N-1}{Nx}\int_0^{x}h(Z+u)du\right],
\ee
and $F(0)$ is interpreted by continuity to be 
\[
E^v\left [\frac{h(Z)}{N} +\frac{(N-1)h(Z)}{N}\right]=E^v[h(Z)].
\]
 Moreover, it follows from the monotonicity of $h$ that $F$ defined via (\ref{e:foc}) is strictly increasing. Also note that (\ref{e:foc}) reduces to (\ref{e:foc1}) when $N=1$. Thus, we shall always refer to (\ref{e:foc}) when discussing the optimal strategies of the insiders regardless of the value of $N$.
 
In what follows  the total informed demand will be denoted by $X$ and we assume following Glosten \cite{Glosten94} that limit prices are given by `tail expectations.' That is, denoting the total demand $X+Z$ by $Y$, 
\be \label{e:h}
h(y)=\left\{\ba{ll}
E[V|Y\geq y], & \mbox{if } y>0;\\
E[V|Y\leq y], & \mbox{if } y<0.
\ea\right\}
\ee
The value $h(0)$ is not relevant for the subsequent computations and can be freely chosen to be any value between the best ask $h(0+):=\lim_{y \downarrow 0} h(y)$ and the best bid $h(0-):=\lim_{y \uparrow 0} h(y)$. 

The definition (\ref{e:h}) entails that liquidity suppliers earn zero aggregate profit on average. Indeed, the expected profit is given by
\bean
E\left[\int_0^Y(h(y)-V)dy\right]&=&E\left[\chf_{[Y>0]}\int_0^Y(h(y)-V)dy+\chf_{[Y<0]}\int_Y^0(V-h(y))dy\right]\\
&=&\int_0^{\infty}E[(h(y)-V)\chf_{[Y\geq y]}]dy+\int_{-\infty}^0E[(V-h(y))\chf_{[Y\leq y]}]dy=0,
\eean
where the last equality is due to the definition of the conditional expectation.

\begin{definition} \label{d:eq} The pair $(h^*,X^*)$ is said to be a Glosten equilibrium if $h^*$ is non-decreasing and non-constant, $X^*\in \bbR$ and
	\begin{itemize}
		\item[i)] $h^*$ satisfies (\ref{e:h}) with $Y=X^*+Z$;
		\item[ii)] $X^*$ is the profit maximising order size for the insider(s) given $h^*$. That is, $V=F(X^*)$, where $F$ is given by (\ref{e:foc}).
	\end{itemize}
\end{definition}
The strict monotonicity of $F$ leads to the following result that in particular yields an explicit formula for the aggregate profit of the insiders.
\begin{proposition}\label{p:wealth}
	Let $(X^*,h)$ be an equilibrium and $F$ defined by (\ref{e:foc}). Then, the following hold:
	\begin{enumerate}
		\item For any $x \in \bbR$ we have
		 \bea
		E^v[h(x+Z)]&=&F(x)+\frac{N(N-1)}{x^N}\int_0^{x}(F(x)-F(y))y^{N-1}dy \nn \\
		&=&F(x)+N(N-1)\int_0^{1}(F(x)-F(xy))y^{N-1}dy \label{e:exph}.
		\eea 
		\item $E^v[h(X^*+Z)]=v$ when $N=1$ and, for $N>1$,  $E^v[h(X^*+Z)]>v$ (resp. $E^v[h(X^*+Z)]<v$) if $X^*>0$ (resp. $X^*<0$). More precisely, 
			\be \label{e:convergence}
		E^v[h(X^*+Z)]=v+N(N-1)\int_0^{1}(v-F(yX^*))y^{N-1}dy.
		\ee
		\item The aggregate expected liquidation profit of the insiders conditional on $V=v$ is given by
		\bea 
		\pi^*(v)&:=&E^v\left[\int_0^{X^*}(v-h(y+Z))dy\right]=N\int_0^{F^{-1}(v)}(v-F(y))\left(\frac{y}{F^{-1}(v)}\right)^{N-1}dy \nn\\
		&=&F^{-1}(v)N\int_0^{1}(v-F(yF^{-1}(v)))y^{N-1}dy.\label{e:wealth}
		\eea
	\end{enumerate}
\end{proposition}

The expression (\ref{e:convergence}) reveals an interesting feature of our model akin to Kyle's model in continuous time. Observe that $h(X^*+Z)$ can be viewed as the last slice that is traded with the limit order traders. Thus, when there is a monopolistic insider, (\ref{e:convergence}) shows that the final slice is priced at the actual value of $V$ similar to the convergence of the equilibrium price to the liquidation value of the asset in the continuous time version of the Kyle model studied in Back \cite{Back} for general payoffs.  As expected, this `convergence' disappears when the liquidation value of the asses is observed by more than one insider. Such a competition leads to a more aggressive trading as in Holden and Subrahmanyam \cite{HolSub92} and results in higher (resp. lower) market valuation if the optimal strategy is to buy (resp. sell).

In practice an important trading benchmark for traders is {\em the implementation shortfall}. Perold \cite{peroldIS} defines it as the difference between a `paper trading' benchmark and the actual trading costs. Assuming that the benchmark is  given by the ex-ante valuation $E[V]$,  the associated shortfall in our context can be defined as follows:
\begin{definition} \label{d:IS}
	Let $h$ be a function that defines the limit order book in a Glosten equilibrium. Then, the implementation shortfall associated with trading $x$ units is given by 
	\[
	IS(x):=\frac{1}{x}\int_0^x E[h(Z+y)]dy.
	\]
\end{definition}
Observe that $IS(x)$ is simply the expected average cost of trading $x$ units. As the following result shows, it is smaller than the marginal cost $F(x)$, which is given by the first order condition (\ref{e:foc}). 
\begin{proposition}
	\label{p:IS} 	Let $h$ be a function that defines the limit order book in a Glosten equilibrium and $F(x)$ be given by the first order condition (\ref{e:foc}). Then,
	\[
	IS(x)=N \int_0^1 F(xy)y^{N-1}dy.
	\]
	In particular, $IS(x)< F(x)$ for $x>0$ and $IS(x)>F(x)$ for $x <0$.
\end{proposition}
 \section{Characterisation of Glosten equilibrium}
Suppose that $(X^*,h^*)$ is an equilibrium. Writing $h$ instead of $h^*$ to ease the exposition and using the definition of $h$ when $y>0$, we get
\[
h(y)=E[V|X^*+Z\geq y]=E[V|F^{-1}(V)\geq y-Z]=E[V|V\geq F(y-Z)].
\]
Similarly, $h(y)=E[V|V\leq F(y-Z)]$ for $y<0$.

Next introduce the {\em right-continuous} functions $\Phi^{\pm}$ and $\Pi^{\pm}$ via
\bean
\Phi^+(y)&:=&E[V\chf_{[V> y]}], \; \Pi^+(y):= P(V> y)\\
\Phi^-(y)&:=&E[V\chf_{[V\leq y]}],\; \Pi^-(y):= P(V\leq y)=1-\Pi^+(y).
\eean
Note that $\Phi^+(y)+\Phi^-(y)=E[V]$ for all $y\in \bbR$. Moreover, $\Phi^+(y-)=E[V\chf_{[V\geq y]}]$ and  $\Pi^+(y-)= P(V\geq  y)$.

Now let us compute $h(y)$ for $y>0$:
\bean
h(y)&=&E[V|V\geq F(y-Z)]= \frac{E[V\chf_{[V\geq F(y-Z)]}]}{P(V\geq F(y-Z))}\nn\\
&=& \frac{\int_{-\infty}^{\infty}\Phi^+(F(y-z)-)q(\sigma,z)dz}{\int_{-\infty}^{\infty}\Pi^+(F(y-z)-)q(\sigma,z)dz}. 
\eean
On the other hand, $\Pi^+(x)\neq \Pi^+(x-)$ at most for countably many $x$. Thus, $\Phi^+(x)= \Phi^+(x-)$ for almost all $x$ and we have for all $y>0$
\be
h(y)=\frac{\int_{-\infty}^{\infty}\Phi^+(F(y-z))q(\sigma,z)dz}{\int_{-\infty}^{\infty}\Pi^+(F(y-z))q(\sigma,z)dz}. \label{e:hplus}
\ee
Similarly, for $y<0$
\be \label{e:hminus}
h(y)=\frac{\int_{-\infty}^{\infty}\Phi^-(F(y-z))q(\sigma,z)dz}{\int_{-\infty}^{\infty}\Pi^-(F(y-z))q(\sigma,z)dz}. 
\ee
In order to obtain an equation for $F$ it will be convenient to define, for any continuous $g$, the mappings
\[
\phi_g^{+}(x):=\frac{\int_{-\infty}^{\infty}\Phi^+(g(z))q(\sigma,x-z)dz}{\int_{-\infty}^{\infty}\Pi^+(g(z))q(\sigma,x-z)dz} \; \mbox{ and }\; \phi_g^{-}(x):=\frac{\int_{-\infty}^{\infty}\Phi^-(g(z)q(\sigma,x-z)dz}{\int_{-\infty}^{\infty}\Pi^-(g(z))q(\sigma,x-z)dz}.
\]
Let us also set
\be \label{d:phig}
\phi_g(x):= \phi_g^{+}(x)\chf_{x\geq 0} + \phi_g^{-}(x)\chf_{x<0}.
\ee

Now, combining (\ref{e:foc}), (\ref{e:hplus}) and (\ref{e:hminus}) yields an equation for $F$:
\be \label{e:F}
F(x)=\frac{1}{N}\int_{-\infty}^{\infty}q(\sigma,x-z)\phi_F(z)dz +\frac{N-1}{N x}\int_0^x dy \int_{-\infty}^{\infty}q(\sigma,y-z)\phi_F(z)dz,
\ee
 If one can change the order of integration in above\footnote{We shall see later that this is justified when $V$ is bounded from below}, then (\ref{e:F}) can be rewritten as
\be \label{e:FN2}
F(x)=\int_{-\infty}^{\infty}\left\{\frac{1}{N}q(\sigma,x-z)+\frac{N-1}{N} \bar{q}(\sigma,x,z)\right\}\phi_F(z)dz,
\ee
with
\be \label{e:qbar}
\bar{q}(\sigma,x,z):=\chf_{x\neq 0}\,\frac{1}{x}\int_0^x \,q(\sigma,y-z)dy+ \chf_{x =0}q(\sigma,z).
\ee
The preceding calculations show that the existence of a solution for the above integral equation is a necessary condition for equilibrium. We in fact have the converse, too.
\begin{theorem} \label{t:eq}
	A Glosten equilibrium exists if and only if there exists a  function $F :\bbR\to \bbR$ that satisfies (\ref{e:F}). Given such a solution $F$, $(X^*, h^*)$ constitutes an equilibrium, where $X^*=F^{-1}(V)$ and $h^*$ is defined via (\ref{e:hplus}) and (\ref{e:hminus}).
\end{theorem}
In view of the above theorem finding an equilibrium boils down to finding a solution of (\ref{e:F}). Moreover, equilibrium will be uniquely defined if there exists a unique solution of (\ref{e:F}). As usual, solutions of (\ref{e:F}) will be identified as fixed-points of a mapping. Before analysing in depth this fixed-point problem, we shal observe few properties of equilibrium. The following scaling property of $F$ is inherited from the analogous property of the Gaussian distribution. 
\begin{proposition} \label{p:scaling}
	Let $F(1;\cdot)$ be a solution of (\ref{e:F}) with $\sigma=1$. Then the function $x \mapsto F(1; \frac{x}{\sigma})$ solves (\ref{e:F}).
\end{proposition}
\begin{proof}
	\bean
	F\Big(1;\frac{x}{\sigma}\Big)&=&\frac{1}{N}\int_{-\infty}^{\infty}q\Big(1,\frac{x}{\sigma}-z\Big)\phi_F(z)dz +\sigma\frac{N-1}{N x}\int_0^{\frac{x}{\sigma}} dy \int_{-\infty}^{\infty}q(1,y-z)\phi_F(z)dz\\
	&=&\frac{1}{N}\int_{-\infty}^{\infty}q\Big(\sigma,x-z\Big)\phi_F\Big(\frac{z}{\sigma}\Big)dz +\frac{N-1}{N x}\int_0^{x} dy \int_{-\infty}^{\infty}q\Big(1,\frac{y}{\sigma}-z\Big)\phi_F(z)dz\\
	&=&\frac{1}{N}\int_{-\infty}^{\infty}q\Big(\sigma,x-z\Big)\phi_F\Big(\frac{z}{\sigma}\Big)dz +\frac{N-1}{N x}\int_0^{x} dy \int_{-\infty}^{\infty}q(\sigma,y-z)\phi_F\Big(\frac{z}{y}\Big)dz.
	\eean
	Similar manipulations yield
	\[
	\phi_F^{\pm}\Big(\frac{z}{y}\Big)=\frac{\int_{-\infty}^{\infty}\Phi^{\pm}\left(F\Big(1;\frac{u}{y}\Big)\right)q\Big(\sigma,u-z\Big)du}{\int_{-\infty}^{\infty}\Pi^{\pm}\left(F\Big(1;\frac{u}{y}\Big)\right)q(\sigma,u-z)du} ,\]
	which establishes the claim.
	\end{proof}
A straightforward corollary to the above is the following.
\begin{corollary} 
	Consider the solutions of (\ref{e:F}) for any $\sigma$.
	\begin{enumerate}
		\item If (\ref{e:F}) has a unique solution for some $\sigma$, it has a unique solution for all $\sigma$s.
		\item If (\ref{e:F}) has a unique solution for some $\sigma$, $F(0)$ does not depend on $\sigma$.
		\item Let $h(\sigma; \cdot)$ be the function defined via (\ref{e:hplus}) and (\ref{e:hminus}), where $F$ is the unique solution of (\ref{e:F}) for the given $\sigma$. Then, $h(\sigma; x)=h(1;\frac{x}{\sigma})$ for all $x \neq 0$. 
		\item Suppose (\ref{e:F}) has a unique solution for some $\sigma$ and let $X^*(\sigma)$ be the optimal order size for the given $\sigma$. Then, $X^*(\sigma)=\sigma X^*(1)$.
	\end{enumerate}
\end{corollary}
In particular, $h$ inherits the same scaling property of $F$. A simple but striking consequence of this property concerns the equilibrium {\em bid-ask spread}.
\begin{definition}
	Let $h$ be a function with right and left limits defining an order book. The bid-ask spread of the associated order book is given by $h(0+)-h(0-)$.
\end{definition} 
\begin{corollary}
	Suppose that uniqueness holds for the solutions of (\ref{e:F}). Then, $h(0+)-h(0-)$ does not depend on $\sigma$. Moreover, for fixed $x$, $h(x)$ is decreasing in $\sigma$ for $x>0$ and increasing in $\sigma$ for $x<0$.
\end{corollary}
Thus, the liquidity suppliers charge a bid-ask spread that does not vanish even if the amount of `noise' trading is excessively large. Moreover, the dependency of the limit price on $\sigma$ is not monotone. In particular, $h(x)$ approaches to the supremum (resp. infimum) of the set of possible values of $V$ for $x>0$ (resp. $x<0$) when $\sigma \rar 0$. We shall see  these phenomena occurring explicitly in Examples \ref{ex:binomial} and \ref{ex:trinomial}. Furthermore, if we consider the limiting behaviour in the other direction as $\sigma \rar \infty$, we observe that the order book gets flatten and converges to the one that yields transaction costs that are proportional to the order size.

A similar scaling property will hold when the signal distribution possesses a form of {\em self-similarity,} which can be proved by similar methods.
\begin{corollary}
	Suppose that uniqueness holds for the solutions of (\ref{e:F}) and  that $V = V(t)$ for some $t>0$ such that the random variables $(V(t))_{t>0}$ are self-similar in the sense that $V(t)\eid t^{\alpha}V(1)$ for some $\alpha>0$ for all $t>0$. Let $F_t$ be the solution of (\ref{e:F}) for $V=V(t)$. Then, $F_t=t^{\alpha}F_1$.  Moreover, the bid-ask spread is increasing in $t$.
\end{corollary}

A typical example of a self-similar random variable is mean-zero Normal random variable. In this case $V(t)\eid t V(1)$, where $t$ is the parameter corresponding to standard deviation. Similarly, $V(t)\eid t V(1)$, when $V(t)$ corresponds to an exponential random variable with mean (hence, standard deviation) $t$.  The above corollary therefore shows that the bid ask spread gets bigger as the variance of the information signal gets higher, indicating that the liquidity suppliers charge a bigger bid-ask spread when the informational asymmetry gets bigger.

Another consequence of the uniqueness of solutions of (3.11) is that the aggregate expected profit of the insiders vanishes as $N \rar \infty$.
\begin{proposition} \label{p:limitprofit}
Suppose that  $-\infty<m<M<\infty$, there exists a unique solution $F_N$ of (\ref{e:F}) for each $N\geq 1$, and uniqueness holds for the solutions of
\be \label{e:Flimit}
F(x)=\frac{1}{x}\int_0^x dy \int_{-\infty}^{\infty}q(\sigma,y-z)\phi_F(z)dz.
\ee
Assume further that $\Pi^+$ is continuous. Then $F_{\infty}=\lim_{N\rar \infty} F_N$ exists and solves (\ref{e:Flimit}). Moreover, $\lim_{N\rar \infty}\pi^*(v)=0$, where $\pi^*$ is the aggregate expected profit as defined in (\ref{e:wealth}).
\end{proposition}
\begin{definition} $V$ is said to be symmetric if $V$ and $-V$ have the same distribution. That is, $\Pi^+(y-)=\Pi^-(-y)$ for all $y$.
\end{definition}
\begin{proposition} \label{p:symmetric}
	Suppose that $V$ is symmetric and there exists a unique solution $F$ to (\ref{e:F}). Then $F$ is symmetric, i.e. $F(x)=-F(-x)$ for all $x$. Moreover, any symmetric solution of (\ref{e:F}) is  also a solution of 
\be	\label{e:Fsym}
F(x)=\frac{1}{N}\int_{0}^{\infty}q_0(\sigma,x-z)\phi^+_F(z)dz +\frac{N-1}{N x}\int_0^x dy \int_{0}^{\infty}q_0(\sigma,y-z)\phi^+_F(z)dz,
\ee
where $q_0(\sigma,x,z):=q(\sigma,x-z)-q(\sigma,x+z)$.
	
\end{proposition}
\begin{proof}
Observe that $-\Phi^+(y-)=E[-V\chf_{[V\geq  y]}]	=E[-V\chf_{[-V\leq  -y]}]=E[V\chf_{[V\leq  -y]}]=\Phi^-(-y)$. Thus, utilising the fact that $\Phi^+$ and $\Pi^+$ differ from their left limits at most at countably many points and $q$ is symmetric around zero, we obtain
\bean
\int_{-\infty}^{\infty}q(\sigma,-x-z)\phi_F(z)dz&=&\int_{-\infty}^0dz\,q(\sigma,x-z)\frac{\int_{-\infty}^{\infty}\Phi^+(F(-u))q(\sigma,z-u)du}{\int_{-\infty}^{\infty}\Pi^+(F(-u))q(\sigma,z-u)du}\\
&+&\int_0^{\infty}dz\,q(\sigma,x-z)\frac{\int_{-\infty}^{\infty}\Phi^-(F(-u))q(\sigma,z-u)du}{\int_{-\infty}^{\infty}\Pi^-(F(-u))q(\sigma,z-u)du}\\
&=&-\int_{-\infty}^0dz\,q(\sigma,x-z)\frac{\int_{-\infty}^{\infty}\Phi^-(-F(-u))q(\sigma,z-u)du}{\int_{-\infty}^{\infty}\Pi^-(-F(-u))q(\sigma,z-u)du}\nn \\
&&-\int_0^{\infty}dz\,q(\sigma,x-z)\frac{\int_{-\infty}^{\infty}\Phi^+(-F(-u))q(\sigma,z-u)du}{\int_{-\infty}^{\infty}\Pi^+(-F(-u))q(\sigma,z-u)du}\\
&=&-\int_{-\infty}^{\infty}q(\sigma,-x-z)\phi_G(z)d,
\eean
where  $G(x):=-F(-x)$. 
Moreover,
\bean
-\frac{1}{ x}\int_0^{-x} dy \int_{-\infty}^{\infty}q(\sigma,y-z)\phi_F(z)dz&=&\frac{1}{ x}\int_0^{x} dy \int_{-\infty}^{\infty}q(\sigma,-y-z)\phi_F(z)dz\\
&=&-\frac{1}{ x}\int_0^{x} dy \int_{-\infty}^{\infty}q(\sigma,y-z)\phi_{G}(z)dz.
\eean
 Thus, $-G$ is also a solution of (\ref{e:F}), which establishes the first assertion. The second assertion follows from a change of variable in (\ref{e:F}) as above and using the assumed symmetry of $F$.
\end{proof}
\begin{example} \label{ex:binomial}
	Suppose that $P(V=1)=P(V=-1)=\half$. Then, the unique symmetric solution of (\ref{e:F}) is defined by
	\[
	F(x)=\frac{1}{N}\int_0^{\infty} q_0(\sigma,x,z)dz+\frac{N-1}{Nx}\int_0^xdy \int_0^{\infty}dz q_0(\sigma,y,z), \quad x\geq 0.
	\]
	
	In this case it is easily seen that in equilibrium $X^*=\infty$ (resp. $x^*=-\infty$) if $V=1$ (resp. $V=-1$). Moreover, $h^*(y)=\chf_{[y>0]}-\chf_{[y<0]}$. 
	
	Although the insiders' optimal market order is to buy or sell an infinite amount, their profit remains finite. Indeed, when $V=1$, the aggregate expected profit is given by
	\bean
 \int_0^\infty E^1(1-h(Z+y))dy&=&2 E^1\left(\int_0^{\infty}\chf_{[Z<-y]}dy\right)=2\int_0^{\infty}P(Z>y)dy\\
	&=&2E[Z^+]=E[|Z|]=\sigma \sqrt{\frac{2}{\pi}},
	\eean
	where $Z^+=\max\{Z,0\}$ and the second line is due to the fact that $P(Z>y)=P(Z^+>y)$. Note that the total profit is independent of $N$.
\end{example}
\begin{example}
	\label{ex:trinomial} Consider the case $P(V=-1)=P(V=0)=P(V=1)=\frac{1}{3}$. Then, similar considerations as above should yield $F(\infty)=1$ and $F(-\infty)=-1$. Moreover, symmetry considerations must lead to $F(0)=0$, which suggests that the insider does not trade when $V=0$. Indeed, the unique solution of (\ref{e:F}) is given by
	\bean
	F(x)&=&\frac{1}{N}\int_0^{\infty}q_0(\sigma,x,z)\frac{1}{1+P(Z\geq z)}dz+ \frac{N-1}{Nx}\int_0^xdy\int_0^{\infty}dz q_0(\sigma,y,z)\frac{1}{1+P(Z\geq z)}.
	\eean
 Differently from Example \ref{ex:binomial}, the order book will not be flat since  the insider does not trade when $V=0$. One can obtain $h^*$ via (\ref{e:hminus}) and (\ref{e:hplus}). Alternatively, for $y>0$
	\bean
	h(y)&=&\frac{E[V\chf_{[X^*(V)+Z\geq y]}]}{P(X^*(V)+Z\geq y)}=\frac{P(V=1)}{P(V=1)+P(V=0,Z\geq y)}=\frac{P(V=1)}{P(V=1)+P(V=0)P(Z\geq y)}\\
	&=&\frac{1}{1+P(Z\geq y)},
\eean
	where the first equality follows from the fact that $X^*(V)$ is infinite when $V=-1$ or $1$, and the 	third follows from the independence of $V$ and $Z$.
	Similarly, for $y<0$, 
	\[
	h(y)=-\frac{1}{1+P(Z\leq y)}.
	\]
	In particular, the bid-ask spread is given by $h(0+)-h(0-)=\frac{4}{3}$, independent of the noise volume and of $N$.
	
\end{example}
\subsection{Existence of Glosten equilibrium}
We shall denote the interior of the support of the random variable $V$ by $(m,M)$, where $-\infty\leq m <M\leq \infty$, and define   on the support of $V$ 
\be \label{d:Psipm}
\Psi^{\pm}(y):=\frac{\Phi^{\pm}(y)}{\Pi^{\pm}(y)}
\ee
so that $\Psi^{+}(y)=E[V|V> y]$ and $\Psi^{-}(y)=E[V|V\leq y]$. 

 We impose the following condition on the function $F$ to ensure that the integral equation (\ref{e:F}) is well-defined and  changing the order of integration in (\ref{e:FN2}) is justified.
\begin{assumption}\label{a:Finteg} 
	\[
	\int_{-\infty}^0\phi^-_F(z)q(\sigma,z)dz>-\infty.
	\]
\end{assumption}
Observe that the above is automatically satisfied if $V$ is bounded from below in view of (\ref{e:phi-bd}). Moreover, if $F$ is a continuous function satisfying (\ref{e:F}) and $V$ is symmetric, it satisfies the above assumption in view of Proposition \ref{p:symmetric} and the finiteness of $F$.

One of the useful consequences of the above assumption is the strict monotonicity of solutions of (\ref{e:F}).

\begin{lemma}\label{l:Flim}
	Let $F$ be a continuous non-decreasing solution of (\ref{e:F}) or (\ref{e:Flimit}) satisfying Assumption \ref{a:Finteg}. Then, $\lim_{x\rar \infty}F(x)=M$ and $\lim_{x\rar -\infty}F(x)=m$. Consequently, $F$ is strictly increasing.
\end{lemma}

\begin{theorem} \label{t:existence} Suppose $-\infty<m<M<\infty$. Then, there exists a Glosten equilibrium.
\end{theorem}

Under the hypotheses of Theorem \ref{t:existence} it is clear that $m<F(x)<M$ for $x \in \bbR$. However, in view of  the bounds given by (\ref{e:phi+bd}) and (\ref{e:phi-bd}) it is possible to obtain sharper bounds on any solution of (\ref{e:F}). 
\begin{theorem}\label{t:comparison} Suppose $-\infty<m<M<\infty$ and let $\Psi^{\pm}$ be as in (\ref{d:Psipm}). Then the following statements are valid.
	\begin{enumerate}
		\item There exists a maximal nondecreasing solution\footnote{$R$ is a maximal nondecreasing solution if $R(x)\geq r(x)$ for all $x \in \bbR$, where $r$ is any other nondecreasing solution.} to
		\bea \label{e:compareup}
		R(x)&=&\int_{0}^{\infty}dz\left\{\frac{1}{N}q(\sigma,x-z)+\frac{N-1}{N} \bar{q}(\sigma,x,z)\right\}\int_{-\infty}^{\infty}\Psi^+(R(y))q(\sigma,z-y)dy\nn \\
		&&+ E[V]\int^{0}_{-\infty}dz\left\{\frac{1}{N}q(\sigma,x-z)+\frac{N-1}{N} \bar{q}(\sigma,x,z)\right\}.
		\eea
		Moreover, this maximal solution is not constant and any solution of (\ref{e:F}) is bounded from above by this maximal solution. 
		\item There exists a minimal nondecreasing solution\footnote{$l$ is a minimal nondecreasing solution if $l(x)\leq L(x)$ for all $x \in \bbR$, where $L$ is any other nondecreasing solution.} to 
		\bea \label{e:comparedown}
		l(x)&=& E[V]\int_{0}^{\infty}dz\left\{\frac{1}{N}q(\sigma,x-z)+\frac{N-1}{N} \bar{q}(\sigma,x,z)\right\} \\
		&&+\int^{0}_{-\infty}dz\left\{\frac{1}{N}q(\sigma,x-z)+\frac{N-1}{N} \bar{q}(\sigma,x,z)\int_{-\infty}^{\infty}\Psi^-(l(y))q(\sigma,z-y)dy\right\}\nn.
		\eea
			Moreover, this minimal solution is not constant and any solution of (\ref{e:F}) is bounded from below  by this minimal solution. 
	\end{enumerate}
\end{theorem}

We shall see in the next section that $F$ behaves like $R$ (resp. $l$) as $x \rar \infty$ (resp. $x\rar -\infty$).
\section{Market impact asymptotics}
In this section we are interested in the {\em market impact} associated with large orders. More precisely, we will be computing the asymptotics of the marginal cost of trades, which is given by the function $F$.  As we shall see later, the asymptotic form of $F$ will coincide (up to a scaling factor) with that of implementation shortfall $IS$. We will also be able to compute the tail asymptotics of the distribution of the total demand in equilibrium. 
\begin{definition}
	A function $g:(0,\infty)\to (0,\infty)$ is said to be {\em regularly varying of index $\rho$} at $\infty$ if 
	\[
	\lim_{\lambda \rar \infty}\frac{g(\lambda x)}{g(\lambda)}=x^{\rho},\quad \forall x>0.
	\]
	Analogously, a function $g:(-\infty,0)\to (0,\infty)$ is said to be  regularly varying of index $\rho$ at $-\infty$ if  $g(-x)$ is regularly varying of index $\rho$ at $\infty$.
\end{definition}
We shall first start with the asymptotics of solutions of (\ref{e:compareup}) and (\ref{e:comparedown}) which will later allow us to compute the asymptotics of interest.
\begin{theorem} \label{t:RV}
	Suppose that  $N>1$, $-\infty<m<M<\infty$, and $\Pi^+$ has a continuous derivative. Then, the following statements are valid:
	\begin{enumerate}
		\item Define $\Psi^+_x(M):=\lim_{x\rar M}\frac{d}{dx}\Psi^+(x)$ and set $G:=M-R$, where $R$ is any nondecreasing solution of (\ref{e:compareup}). Then,  $G$ is regularly varying of index $\rho^+$ at $\infty$, where
		\be \label{e:RVindex+}
		\rho^+=\frac{\Psi^+_x(M)-1}{1-\frac{\Psi^+_x(M)}{N}}.
		\ee
		\item Define  $\Psi^-_x(m):=\lim_{x\rar m}\frac{d}{dx}\Psi^-(x)$ and set $G:=l-m$, where $l$ is any nondecreasing solution of (\ref{e:comparedown}). Then,  $G$ is regularly varying of index $\rho^-$ at $-\infty$, where
		\be \label{e:RVindex-}
		\rho^-=\frac{\Psi^-_x(m)-1}{1-\frac{\Psi^-_x(m)}{N}}.
		\ee
	\end{enumerate}
\end{theorem}
\begin{remark}
	Observe that when $M$ is finite,
	\[
	\Psi^+_x(M)=\lim_{x\rar M}\frac{M-\Psi^+(x)}{M-x}\leq \lim_{x\rar M}\frac{M-x}{M-x}=1.
	\]
	Similarly, $\Psi^-_x(m)\leq 1$ if $m$ is finite.
\end{remark}
The above shows that $R$ (resp. $l$) is slowly varying at $\infty$ (resp. $-\infty$) if  $\Psi^+_x(M)=1$ (resp. $\Psi^-_x(M)=1$). We can obtain a better estimate of how slow its variation is under a further assumption on $\Psi^{\pm}$. 

\begin{theorem} \label{t:SVlog}
	Assume that  $N>1$, $-\infty<m<M<\infty$, $\Pi^+$ has a continuous derivative and $\Psi^+_x(M)$ and $\Psi^-_x(m)$ are as in Theorem \ref{t:RV}. Then, the following statements are valid:
\begin{enumerate}
		\item Suppose that $\Psi^+_x(M)=1$ and there exist an integer $n\geq 1$ and a real constant $k\in (0,\infty)$ such that
		\[	\lim_{x\rar M} \frac{\Psi^+(x)-x}{(M-x)^{n+1}}=\frac{1}{k}.\]
Then, the following asymptotics\footnote{We write $f(x)\sim g(x)$, $x\rar\pm \infty$ if $\lim_{x \rar \pm \infty}\frac{f(x)}{g(x)}=1$.} hold:
	\be \label{e:SVlog+}
M-R(x) \sim \left(\frac{N}{N-1}\frac{n}{k}\right)^{-\frac{1}{n}}(\log x)^{-\frac{1}{n}}, \quad x \rar \infty,
	\ee
	where $R$ is any solution of (\ref{e:compareup}).
	\item Suppose that   $\Psi^-_x(m)=1$ and there exist an integer $n\geq 1$ and a real constant $k\in (0,\infty)$ such that
	\[
	\lim_{x\rar m} \frac{x-\Psi^-(x)}{(x-m)^{n+1}}=\frac{1}{k}.
	\]
	Then
	\be \label{e:SVlog-}
	l(x)-m \sim \left(\frac{N}{N-1}\frac{n}{k}\right)^{-\frac{1}{n}}(\log |x|)^{-\frac{1}{n}}, \quad x \rar -\infty,
	\ee
		where $l$ is any solution of (\ref{e:comparedown}).
		\end{enumerate}
\end{theorem}
\begin{corollary}  \label{c:SVlog}
	Assume that  $N>1$, $-\infty<m<M<\infty$, $\Pi^+$ has a continuous derivative and $\Psi^+_x(M)$ and $\Psi^-_x(m)$ are as in Theorem \ref{t:RV}. 
	 Then, the following statements are valid:
	\begin{enumerate}
		\item Suppose that there exist an integer $n\geq 0$ and a real constant $k\in (0,\infty)$ such that
		\[	\lim_{x\rar M} \frac{\Psi^+(x)-x}{(M-x)^{n+1}}=\frac{1}{k}.\]
		Then, $\Pi^+(R)$ is regularly varying at $\infty$ of index 
		\[
		-\frac{\Psi^+_x(M)}{1-\frac{\Psi^+_x(M)}{N}},
		\]
		where $R$ is any solution of (\ref{e:compareup}).
		\item Suppose that there exist an integer $n\geq 0$ and a real constant $k\in (0,\infty)$ such that
		\[	\lim_{x\rar M} \frac{x-\Psi^-(x)}{(x-m)^{n+1}}=\frac{1}{k}.
		\]
		Then, $\Pi^-(l)$ is regularly varying at $-\infty$ of index 
		\[
		-\frac{\Psi^-_x(m)}{1-\frac{\Psi^-_x(m)}{N}},
		\]
		where $l$ is any solution of (\ref{e:comparedown}).
	\end{enumerate}
\end{corollary}
\begin{remark} \label{r:remarkn}Note that if $\Psi^+_x(M)\neq 1$, we have
	\[
	\lim_{x \rar M}\frac{\Psi^+(x)-x}{M-x}=1-\Psi^+_x(M),
	\] 
	implying $n=0$ and $k^{-1}=1-\Psi^+(M)$ in Part 1) of Corollary \ref{c:SVlog}. An analogous consideration applies to the second part.
\end{remark}
The results above now allow us to compute the asymptotics of solutions of (\ref{e:F}).
\begin{theorem} \label{t:asympF}
	Assume that  $N>1$, $-\infty<m<M<\infty$, $\Pi^+$ has a continuous derivative and $\Psi^+_x(M)$ and $\Psi^-_x(m)$ are as in Theorem \ref{t:RV}. Let $F$ be any solution of (\ref{e:F}). Then, the following statements are valid:
\begin{enumerate}
	\item $M-F$ is regularly varying of index $\rho^+$ at $\infty$, where $\rho^+$ is given by (\ref{e:RVindex+}). 

Moreover,  if $\Psi^+_x(M)=1$ and there exist an integer $n\geq 1$ and a real constant $k\in (0,\infty)$ such that
	\[	\lim_{x\rar M} \frac{\Psi^+(x)-x}{(M-x)^{n+1}}=\frac{1}{k},
	\]
	 the following asymptotics hold:
	\be \label{e:SVlog+F}
	M-F(x) \sim \left(\frac{N}{N-1}\frac{n}{k}\right)^{-\frac{1}{n}}(\log x)^{-\frac{1}{n}}, \quad x \rar \infty.
	\ee
	\item  $F-m$ is regularly varying of index $\rho^-$ at $-\infty$, where $\rho^-$ is given by (\ref{e:RVindex-}).

Moreover, if   $\Psi^-_x(m)=1$, and there exist an integer $n\geq 1$ and a real constant $k\in (0,\infty)$ such that
	\[
	\lim_{x\rar m} \frac{x-\Psi^-(x)}{(x-m)^{n+1}}=\frac{1}{k},
	\]
	then
	\be \label{e:SVlog-F}
	F(x)-m \sim \left(\frac{N}{N-1}\frac{n}{k}\right)^{-\frac{1}{n}}(\log |x|)^{-\frac{1}{n}}, \quad x \rar -\infty.
	\ee
\end{enumerate}
\end{theorem}
\begin{remark}
Upon integrating by parts we arrive at
\[
\frac{\Psi^+(x)-x}{(M-x)^{n+1}}=\frac{\int_x^M \Pi^+(y)dy}{\Pi^+(x) (M-x)^{n+1}}.
\]
Now suppose that $\Pi^+(x)=\int_x^M p(y)dy$ for some differentiable $p$ such that 
\[
-\frac{p(x)}{p'(x)}\sim \frac{1}{k} (M-x)^{n+1} \mbox{ as } x \rar M,
\]
for some $k>0$ and $n\geq 1$. A straightforward application of L'Hospital rule shows that
\[
\lim_{x \rar M}\frac{\Pi^+(x)}{p(x)(M-x)^{n+1}}=\frac{1}{k},
	\]
	which in turn implies
	\[
	\lim_{x\rar \infty} \frac{\int_x^M \Pi^+(y)dy}{\Pi^+(x) (M-x)^{n+1}}=\frac{1}{k}.
	\]
For instance, if
\[
p(x)\propto (M-x)^{\beta}\exp\left(-\frac{\Sigma}{(M-x)^n}\right),
\]
for some $\beta \in \bbR$, $\Sigma>0$,  and $n\geq 1 $, we deduce that
\[
\lim_{x \rar M}\frac{\Psi^+(x)-x}{(M-x)^{n+1}}=\frac{1}{n\Sigma }.
\]
\end{remark}
We also have the exact analogue of Corollary \ref{c:SVlog} that can be proven by exactly the same arguments.
\begin{corollary}  \label{c:Voldist}
	Assume that  $N>1$, $-\infty<m<M<\infty$, $\Pi^+$ has a continuous derivative and $\Psi^+_x(M)$ and $\Psi^-_x(m)$ are as in Theorem \ref{t:RV}. Let $F$ be any solution of (\ref{e:F}).
	Then, the following statements are valid:
	\begin{enumerate}
		\item Suppose that there exist an integer $n\geq 0$ and a real constant $k\in (0,\infty)$ such that
		\[	\lim_{x\rar M} \frac{\Psi^+(x)-x}{(M-x)^{n+1}}=\frac{1}{k}.\]
		Then, $\Pi^+(F)$ is regularly varying at $\infty$ of index 
		\[
		-\frac{\Psi^+_x(M)}{1-\frac{\Psi^+_x(M)}{N}}.
		\]
		\item Suppose that there exist an integer $n\geq 0$ and a real constant $k\in (0,\infty)$ such that
		\[	\lim_{x\rar M} \frac{x-\Psi^-(x)}{(x-m)^{n+1}}=\frac{1}{k}.
		\]
		Then, $\Pi^-(F)$ is regularly varying at $-\infty$ of index 
		\[
		-\frac{\Psi^-_x(m)}{1-\frac{\Psi^-_x(m)}{N}}.
		\]
	\end{enumerate}
\end{corollary}

\begin{remark}
	\label{r:Voldist} Although Corollary \ref{c:Voldist} appears rather technical, it uncovers the distribution of the total volume traded in equilibrium. Indeed, for $x>0$
	\[
	P(X^*>x)=P(F^{-1}(V)>x)=P(V>F(x))=\Pi^+(F(x)).
	\]  
	Thus, under the hypothesis of Corollary \ref{c:Voldist}, the tail distribution of equilibrium $X^*$ is regularly varying at infinity. That is, 
	\[
	P(X^*>x)=x^{-\zeta^+}s(x),
	\]
	where $s$ is a slowly varying function and
	\[
	\zeta^+:=\frac{\Psi^+_x(M)}{1-\frac{\Psi^+_x(M)}{N}}
	\]
	
	Moreover, since the aggregate order is given by $Y^*=X^*+Z$ and $Z$ and $V$ are independent, we have for $y>0$
	\[
	P(Y^*>y)=\int_{-\infty}^{\infty}dzP(X^*>y-z)q(\sigma, z)=\int_{-\infty}^{\infty}dzP(X^*>z)q(\sigma, y-z),
	\]
	which can easily be shown to be regularly varying at infinity with the same index. Thus,
\be \label{e:Volasympdist+}
P(Y^*>y)=y^{-\zeta^+}s(y), \quad y >0,
\ee
for some regularly varying $s$. In particular, if $V$ has light tails, i.e. $\Psi^+_x(M)=1$, $P(Y^*>y)$ is regularly varying of index $-\frac{N}{N-1}$.

 Analogous computations yield for the sell side
\be \label{e:Volasympdist-}
P(Y^*<-y)=y^{-\zeta^-}s(y), \quad y >0,
\ee
where
\[
\zeta^-:=\frac{\Psi^-_x(m)}{1-\frac{\Psi^-_x(m)}{N}}.
\]

\end{remark}
We have seen in Proposition \ref{p:IS} that the implementation shortfall is smaller than $F$. In view of Theorem \ref{t:asympF} we have a more precise relationship for large $x$.
\begin{corollary} \label{c:ISasymptotics}
Assume that  $N>1$, $-\infty<m<M<\infty$, $\Pi^+$ has a continuous derivative, and let $(h^*,X^*)$ be an equilibrium. Suppose that  $F^*$ is given by (\ref{e:foc}), with $h$ being replaced by $h^*$. Then,
\bea
M-IS^*(x)&\sim& \frac{N}{N+\rho^+}(M-F^*(x)), \qquad x \rar \infty,\\
IS^*(x)-m&\sim& \frac{N}{N+\rho^-}(F^*(x)-m), \qquad x \rar -\infty,
\eea
where $\rho^+$ and $\rho^-$ are as in Theorem \ref{t:RV}.
\end{corollary}
Corllary \ref{c:ISasymptotics} shows that for large $N$ implementation shortfall and the marginal cost of trading are almost indistinguishable. This is even more pronounced when $M-F$ (resp. $F-m$) is slowly varying, i.e. $\rho^+=0$ (resp. $\rho^-=0$). Observe that $IS$ can be estimated from market data using the open limit order book governed by $h$, while the computation of $F$ depends also on the estimate of $N$.
\begin{remark}
	We have focused in this section the asymptotics of $F$ and $IS$. However, one also naturally wonders the asymptotic shape of the limit order book. Using our framework it is easy to show that $F$ and $h$ behave similarly for large values.  Indeed, recalling the measure $\nu$ defined in (\ref{e:measnu}), we   formally  obtain
	\bean
	\lim_{x \rar \infty}\frac{M-h( x)}{M-F(x)}&=&\lim_{x \rar \infty}\int_{-\infty}^{\infty}\nu(x,1,dy)\frac{M-\Psi^+(F(xy))}{M-F(x)}\\
	&=&\Psi^+_x(M) 1^{\rho^+}=\Psi^+_x(M),
	\eean
	using the mean value theorem and the continuity of the derivative of $\Psi^+$ together with the fact that the $F$ is regularly varying with index $\rho^+$ since the measure $\nu(x,1,dy)$ converges to the point mass at $1$ as $x \rar \infty$. In particular, the order book $h$ is also regularly varying with the same index $\rho^+$. 
\end{remark}
\section{Numerical studies}
This section is devoted to description of results obtained in previous version via a naive numerical search for a fixed point: Starting with an $F_0$, we compute $F_{n+1}=TF_n$ until the distance between successive iterations become negligibly small, where  $TF$ corresponds to the right side of (\ref{e:F}). Although the proofs of the statements concerning the existence of equilibrium and its asymptotics relied on a boundedness assumption, we shall also presents the solutions of the fixed point problems associated with equilibria with unbounded signals. 

Since $\sigma$ can be absorbed into the units for measuring equity in view of (\ref{p:scaling}),  we will set $\sigma=1$ in all numerical tests with no loss of generality in this section. 
	
The iteration converges exponentially for all cases considered below. We illustrate the convergence in Figure ~\ref{fig:Convergence} by showing the uniform distance $d(g_n, g_{n+1})$ between successive iterations for log-normal signals with $N=25$. 
\begin{figure}[h]
\includegraphics[scale=0.7, angle=0]{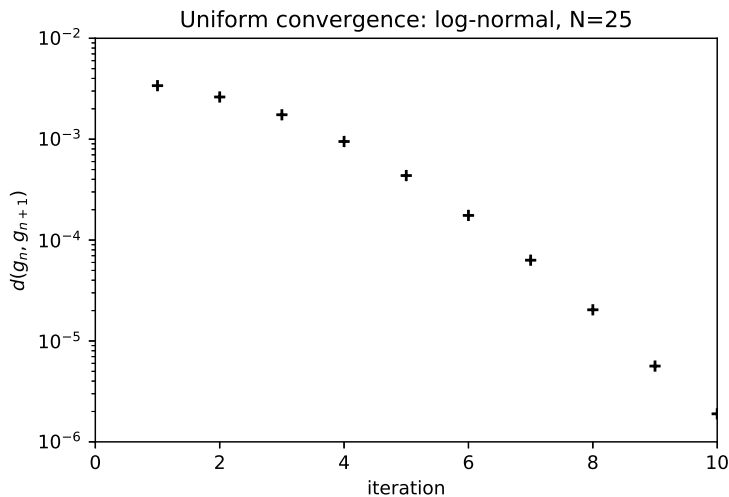}
\caption{The uniform distance between successive iterations falls off exponentially with the iteration number}
\centering
\label{fig:Convergence}
\end{figure}

In practice, the signals distribution is at least partly known to practitioners since it can be inferred from the options montage and the event calendar. For example, an insider may have advance knowledge of the result of clinical trials of a new treatment from a biopharmaceutical company, or the insider may have been tipped about a number to be released in a scheduled earnings call. The signals distributions in these two examples are quite different. In the case of clinical trial results, the signal is drawn from a bounded distribution since the outcome is bounded by total success and total failure. On the other hand, earnings surprises can be arbitrarily large. Earnings calls can be priced in a naive jump-diffusion model by assuming the surprise is drawn from a log-normal distribution; the jump variance is estimated by comparing the at-the-money implied volatility(ATMIV) for the two nearest options expirations, as proposed by Dubinsky and Johannes \cite{Dubinsky2006}.

The shape of market impact as a function of trade size is important to practitioners to address capacity and position sizing. Various functional forms have been explored in the literature (see, e.g.,  \cite{Bershova2013},  \cite{Gabaix2006}, \cite{PBimpact} and \cite{Torre97}), including $\sqrt{X}$ and $\log(1+a X)$. Rejecting either of these forms empirically is challenging due to three practical difficulties: (1) the spread and market impact terms are collinear for small orders, (2) signal-to-noise ratios are weak for executions that take a small percentage of market volume, and (3) for very large trades, order sizes are often increased if liquidity is available, or reduced if liquidity is hard to find, leading to bias in the data. In absence of clear empirical evidence, theoretical predictions for the shape of market impact provide valuable insight into how trading costs scale with trade size. One testable conclusion from this paper is that market impact should be different ahead of an event with an unusual distribution of signals, if market makers believe that an informed trader may have advance knowledge of the signal. 

We consider first the case of bounded signals that was the focus of the previous section.

\subsection{Bounded signals}

\subsubsection{Truncated Gaussian distribution}
\

If signals are drawn from the truncated Gaussian distribution with density $p(v) = {1 \over{erf({{M}\over{\sqrt{2\Sigma}}}) \sqrt{2\pi\Sigma}}} e^{-{{v^2}\over{2\Sigma}}}$ for $v\in[-M,M]$, the  numerical solution for equilibrium $F$ converges to the upper bound as $M-F(x) \sim 1/x^{{{N-1}\over{2N}}}$ in accordance with the predictions of  Theorem \ref{t:asympF}. We show $F$ and the theoretical prediction for its asymptotic behavior in Figure ~\ref{fig:TruncatedGaussianConvergenceToMax}.

\begin{figure}[h]
\includegraphics[scale=0.7, angle=0]{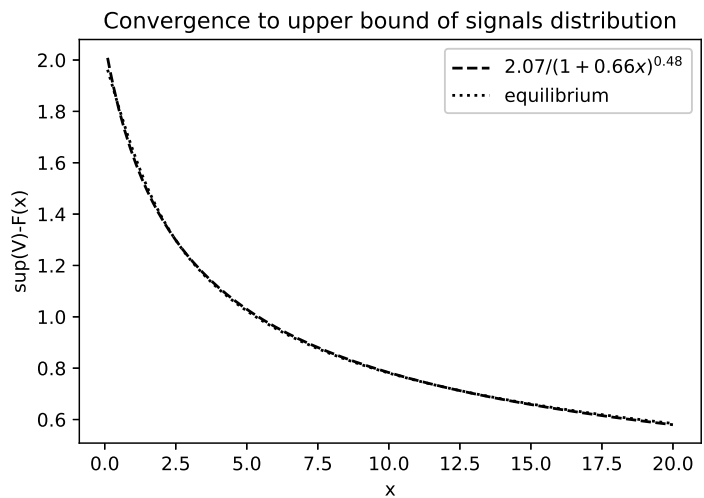}
\caption{The asymptotic behavior of $F$ is shown for the case where signals are drawn from a truncated Gaussian distribution}
\centering
\label{fig:TruncatedGaussianConvergenceToMax}
\end{figure}

\subsubsection{Logit-normal distribution}
\

If price is a probablility-weighted average over two possible outcomes $v_{\pm}=p_0\pm 1$, where the probability is a sigmoidal function of a Gaussian generator potential $g$ as $p={1\over{1+e^{-g}}}$, the signals distribution is the logit-normal distribution with density $p(v)=\sqrt{{2}\over{\pi \Sigma}} {{e^{-{{ln^2({{1-v} \over {1+v}})} \over {2\Sigma}}}} \over {1-v^2}} $. This distribution has support in $[-1,1]$. 

We show the equlibrium solution for $F$, the  order book $h$ and the implementation shortfall for logit-normal signals, in Figure ~\ref{fig:BGaussian}.

\begin{figure}[h]
\includegraphics[scale=0.7, angle=0]{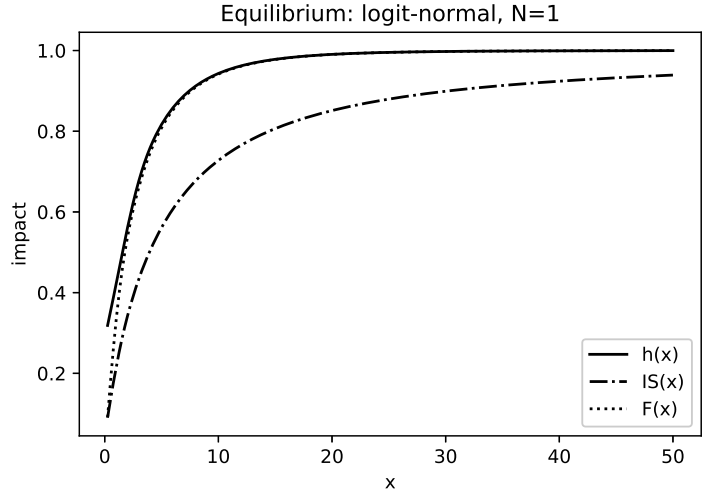}
\caption{Equilibrium impact and shortfall for logit-normal signals for the case of an insider ($N=1$)}
\centering
\label{fig:BGaussian}
\end{figure} 

Note that the logit normal distribution does not satisfy the hypothesis of Theorem \ref{t:asympF}. Thus, our theory cannot predict the asymptotics of $F$ for this distribution. However, the following formal arguments yield the asymptotics that seem to be verified by the numerical experiments.  Observe that for $x>>\sigma$, \ref{e:foc} becomes 
\bea \label{e:equilibriumGaussNoise}
	F(x) \approx {1\over N} h(x) + {{N-1}\over {N x}} \int_0^x h(u) du
\eea

It follows that $F(x)+x F'(x) \approx h(x)+{1\over N} x h'(x)$. Moreover, for large values of $x$, we roughly have $h(x)\approx \Psi^+(F(x))$. Let us consider the large $N$ limit and drop the $1/N$ term. Using the approximation $erfc(x) \approx {{e^{-x^2}} \over {x \sqrt{\pi}}} (1 - {1 \over {2 x^2}})$ and expanding to first order in $1/\log(x)$ we find that asymptotically
\[
xF'=\frac{\Sigma (1-F)}{log(1-F)},
\]
whihc yields $F \to 1 - e^{-k \sqrt{\log(x)}}$

The numerical solution shown in Figure ~\ref{fig:BoundedGaussianConvergenceToMax} for $N=25$ is consistent with this asymptotic form. 

\begin{figure}[h]
\includegraphics[scale=0.7, angle=0]{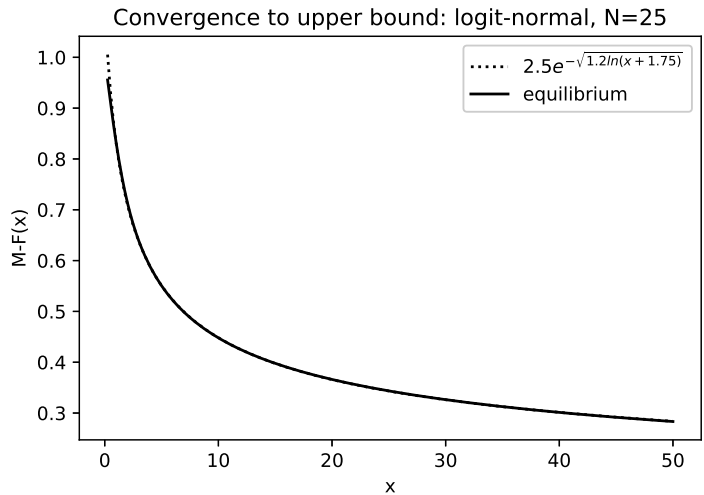}
\caption{Convergence to the upper bound for $M=1$ when signals are drawn from a logit-normal distribution and shared with $N=25$ traders.}
\centering
\label{fig:BoundedGaussianConvergenceToMax}
\end{figure} 

\subsection{Unbounded signals}

In what follows, we place emphasis on unbounded signals; that is, when the support of $V$ is unbounded. Although we do not have a theoretical justification for the existence of a solution for (\ref{e:F}) in the case of unbounded signals, we were able to arrive at numerical solutions via the above numerical search. 

The asymptotic behavior of $F(x)$ for large signals will depend on the tail behaviour of the distribution of $V$. Assuming the interchange of limits and integrals in the proof of Theorem \ref{t:asympF}, we can show that
\[
\gamma(x)= \frac{\Psi^+_x(\infty)}{N}\gamma(x) + \frac{N-1}{Nx}\Psi^+_x(\infty)\int_0^x \gamma(y)dy, \quad x>0,
\]
where $\gamma(x):=\lim_{\alpha \rar \infty}\frac{F(\alpha x)}{F(\alpha)}$ in case of $M=\infty$. Solution of the above equation immediately yields that, for $N>1$, $F$ is regularly varying at $\infty$ of order $\rho^+$, where
\be \label{e:rho+inf}
\rho^+=\frac{\Psi^+_x(\infty)-1}{1-\frac{\Psi^+_x(\infty)}{N}}.
\ee
Observe that when $M=\infty$, $\Psi^+_x(\infty)\geq 1$ in contrast to the bounded case, where $\Psi^+_x(M)\leq 1$. Since $\rho^+$ must be non-negative, this places the restriction on $N$:
\be \label{e:Nconstraint}
N\geq \Psi^+_x(\infty)
\ee
Thus, we conjecture that (\ref{e:Nconstraint}) is a necessary condition for the existence of equilibrium when $M=\infty$.  Observe that for a fat tailed unbounded distribution,  $\Psi^+_x(\infty)>1$. Thus, a sufficient competition among insiders is necessary for the equilibrium to exist. Such a condition is always satisfied in the bounded case since $\Psi^+_x(M)\leq 1$ for $M<\infty$.

As in the bounded case $F$ will be slowly varying at infinity when $\Psi^+_x(\infty)=1$. In this case, if we assume
\be \label{e:Psi+unbd}
\lim_{x\rar \infty} (\Psi^+(x)-x)x^{n-1}=\frac{1}{k}
\ee
for some $k>0$ and $n\geq 1$, formal calculations yield
\be \label{e:Flogunbd}
F(x)\sim \left(\frac{N}{N-1}\frac{n}{k}\right)^{\frac{1}{n}}(\log x)^{\frac{1}{n}}, \quad x\rar \infty.
		\ee

Tables \ref{T:powerlaw} and \ref{T:loglaw} summarise the predicted asymptotics for a class of distributions commonly used in the literature and practice. 
		
\begin{table}[tb]
	\caption{Distributions with power-law impact}
	\label{T:powerlaw}\vspace{1mm}
	\par
	\begin{center}
		
		\bgroup
		\def\arraystretch{1.5}
		
		\begin{tabular}{cccc}
			 Distribution & Density & $\rho^+$ \\
			\hline
			Beta prime & $x^{\lambda-1}(1+x)^{-(\lambda +\alpha)}$	 &$\left(\frac{N-1}{N}\alpha -1\right)^{-1}$\\
			\hline
			Fr\'echet & $(x-\beta)^{-(1+\alpha)}\exp\left\{-\left(\frac{x-\beta}{s}\right)^{-\alpha}\right\}$ &$\left(\frac{N-1}{N}\alpha -1\right)^{-1}$\\
			\hline
			Lomax &$\left(1+\frac{x}{\lambda}\right)^{-(\alpha+1)}$	 &$\left(\frac{N-1}{N}\alpha -1\right)^{-1}$\\
			\hline
			Pareto &$x^{-(\alpha+1)}$ &$\left(\frac{N-1}{N}\alpha -1\right)^{-1}$\\
			\hline
		Student &$\left(1+\frac{x^2}{\alpha}\right)^{-(\alpha+1)/2}$ &$\left(\frac{N-1}{N}\alpha -1\right)^{-1}$\\
		\hline
		\end{tabular}
		
		\egroup
		
	\end{center}
	\par
	\begin{spacing}{1.0}
		\footnotesize In above probability densities are given up to a scaling factor and implicit constraints are enforced to ensure they are well defined with finite mean.  Moreover, $N>\frac{\alpha}{\alpha-1}$ in all of the above. 
	\end{spacing}\vspace{2mm}
\end{table}
\begin{table}[tb]
	\caption{Distributions with logarithmic impact}
	\label{T:loglaw}\vspace{1mm}
	\par
	\begin{center}
		
		\bgroup
		\def\arraystretch{1.5}
		
		\begin{tabular}{ccc}
			Distribution & Density  & Asymptotics \\
			\hline
			Exponential& $\exp(-\lambda x)$	& $\frac{N}{\lambda(N-1)}\log x$\\
			\hline
			Gaussian& $\exp(-(x-\mu)^2/\Sigma)$	& $\sqrt{\frac{2\Sigma N}{N-1}}\sqrt{\log x}$\\
			\hline
		Inverse Gaussian & $x^{-3/2}\exp\left(-\frac{\lambda (x-\mu)^2}{2\mu^2 x}\right)$	& $\frac{2N\mu^2}{\lambda(N-1)}\log x$\\
			\hline
			Normal Inverse Gaussian &$\frac{K_1(\lambda \zeta(x))}{\pi \zeta(x)}\exp(\delta \gamma+\beta(x-\mu)$ &$\frac{N}{(N-1)(\lambda+\beta-1)}\log x$\\
			\hline
			Weibull& $x^{d-1}\exp(-\lambda^p x^p)$	& $\left(\frac{N}{\lambda^p(N-1)}\right)^{1/p}(\log x)^{1/p}$\\
			\hline
		\end{tabular}
		
		\egroup
		
	\end{center}
	\par
	\begin{spacing}{1.0}
		\footnotesize In above probability densities are given up to a scaling factor and implicit constraints are enforced to ensure they are well defined with finite mean.  Moreover, $\zeta(x):=\delta^2 +(x-\mu)^2$ for the Normal Inverse Gaussian distributiuon.
	\end{spacing}\vspace{2mm}
\end{table}

\subsubsection{Gaussian signals}

We assume that the mean of $V$ equals $0$. The numerical solution is shown in Figure ~\ref{fig:GaussianSolution25} together with the $\sqrt{\log{x}}$ asymptotic behavior.

\begin{figure}[h]
	\includegraphics[scale=0.7, angle=0]{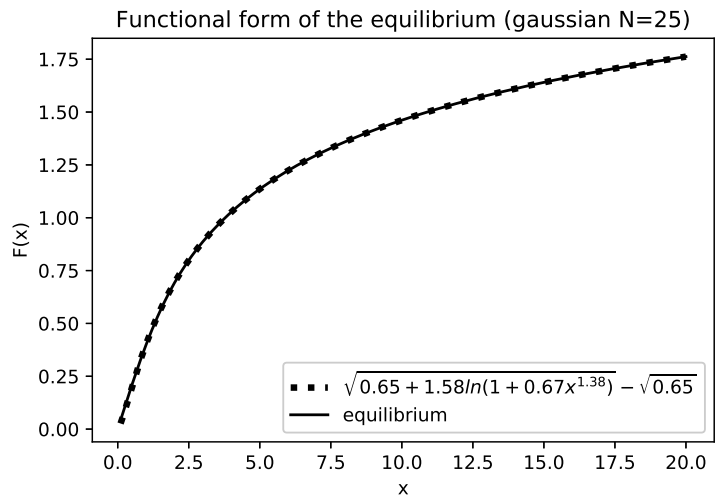}
	\caption{Functional form of the equilibrium for Gaussian signals, for $N=25$}
	\centering
	\label{fig:GaussianSolution25}
\end{figure}

\subsubsection{Log-normal signals}

For a log-normal distribution, the mean is an arbitrary scale factor which we set to $1$ and, thus, the density is $p(v)={{1}\over{\sqrt{2\pi\Sigma}v}} e^{-{{(ln(v)+\Sigma/2)^2} \over {2\Sigma}}}$. We choose a large signal variance $\sqrt{\Sigma}=10\%$ in our numerical experiments below, illustrative of an earnings announcement for a high-volatility name. Moreover, we translate the distribution by $1$ so that the mean is $0$. The equilibrium solutions for $h,F$ and $IS$  are shown for various values of $N$ in Figure ~\ref{fig:LogNormalSolutions}. 

\begin{figure}[h]
\includegraphics[scale=0.7, angle=0]{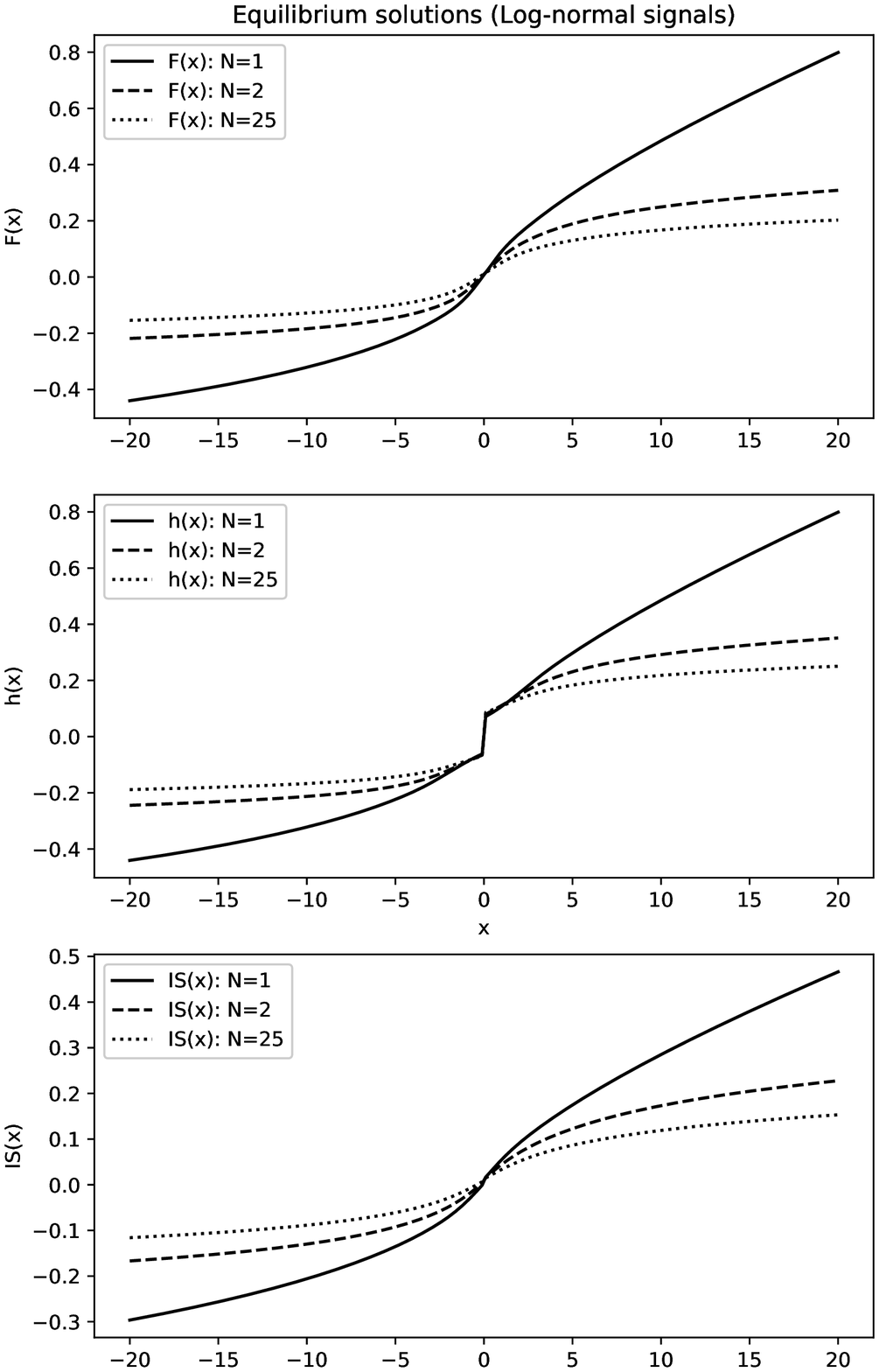}
\caption{Equilibrium solutions for log-normal signals for the cases of an insider ($N=1$), and shared signals with $N=2$, $N=25$}
\centering
\label{fig:LogNormalSolutions}
\end{figure}

Observe that the equilibrium $F$ is not symmetric around its mean: a market maker who is short the stock runs the risk of unbounded losses, whereas for a long position the maximum loss is always bounded since price cannot fall below zero. 

The log-normal distribution does not satisfy the conditions of Theorem \ref{t:asympF}. Thus, we do not have a theoretical prediction for the asymptotic market impact. However,  we find numerically that the implementation shortfall fits $\sqrt{\log(1+a x^{\delta})}$. Asymptotically, both the shortfall and $F$ are  consistent with $\sqrt{\log(x)}$. We note that the large-size behavior $\sqrt{log(x)}$ is more concave than both the square root commonly used by practitioners and also the $log(1+x)$ model that has been suggested in some empirical studies, e.g. \cite{Bershova2013}.
 
\begin{figure}[h]
\includegraphics[scale=0.7, angle=0]{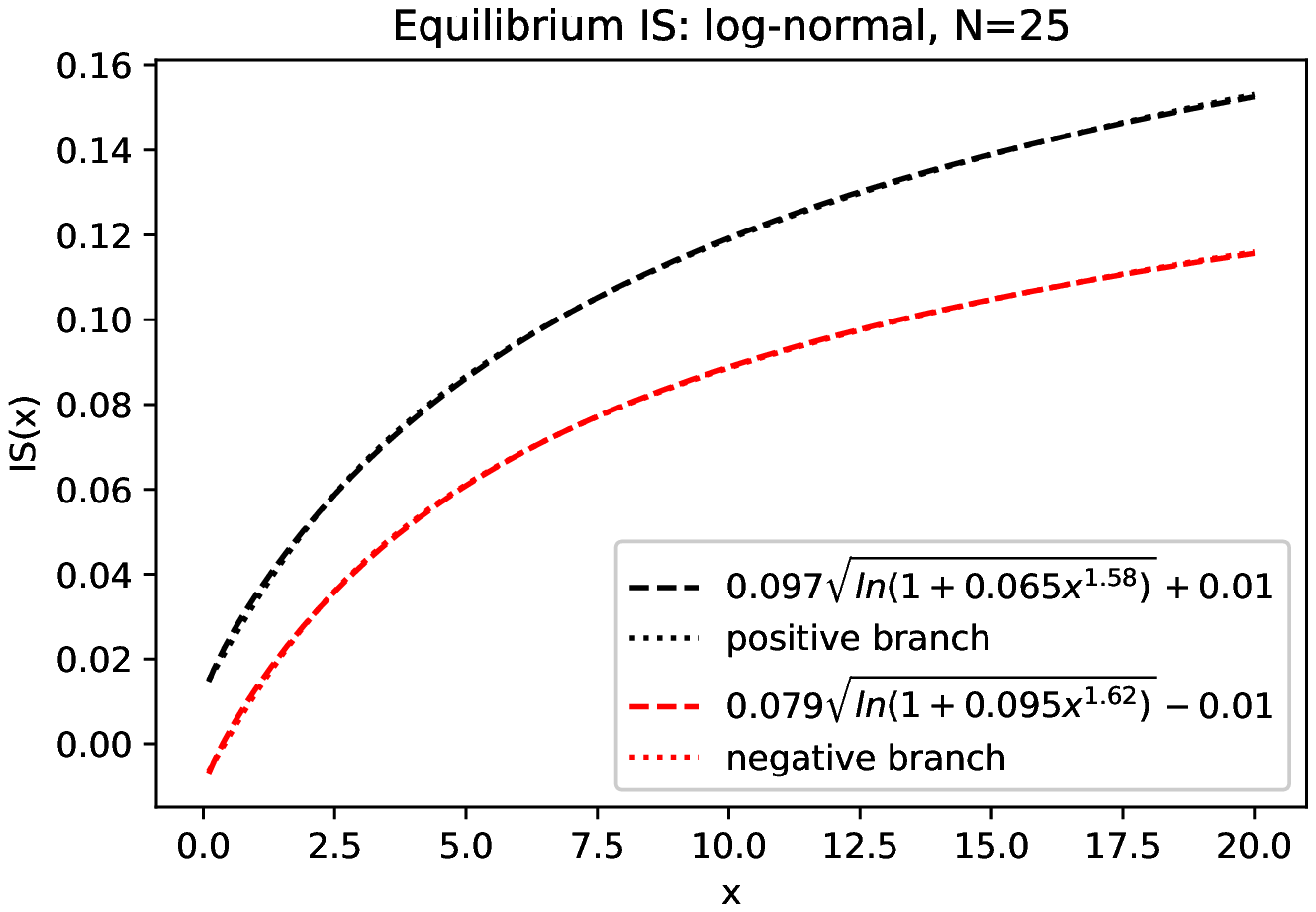}
\caption{Functional form of the equilibrium for Log-normal signals, for $N=25$. The log-normal distribution is not symmetric and this results in a notable difference between the positive and negative branches for cost as a function of trade size.}
\centering
\end{figure}

The asymmetry between positive and negative signals is clearly visible in the figure as we chose a rather large signal with a standard deviation of $10\%$. 

The aggregate profit is a decreasing function of the number of informed investors and we show the profit as a function of trade size below for various values of $N$ in Figure ~\ref{fig:LognormalSolutionProfit}. Moreover, Figure ~\ref{fig:LognormalSpreadVsN} shows how the spread depends on the number of informed investors. 

\begin{figure}[h]
\includegraphics[scale=0.7, angle=0]{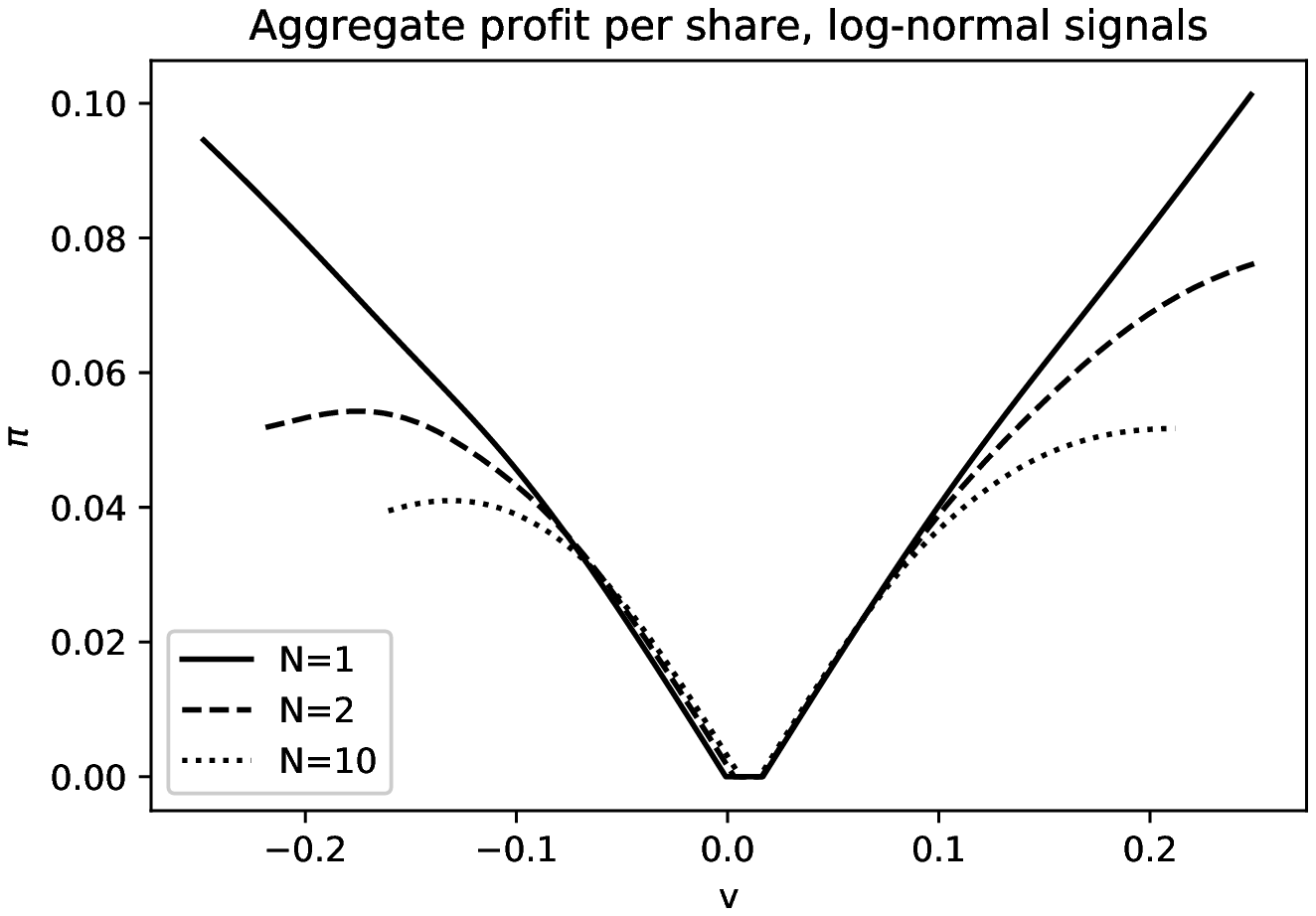}
\caption{Aggregate investor profit per share for Log-normal signals}
\centering
\label{fig:LognormalSolutionProfit}
\end{figure}

\begin{figure}[h]
\includegraphics[scale=0.7, angle=0]{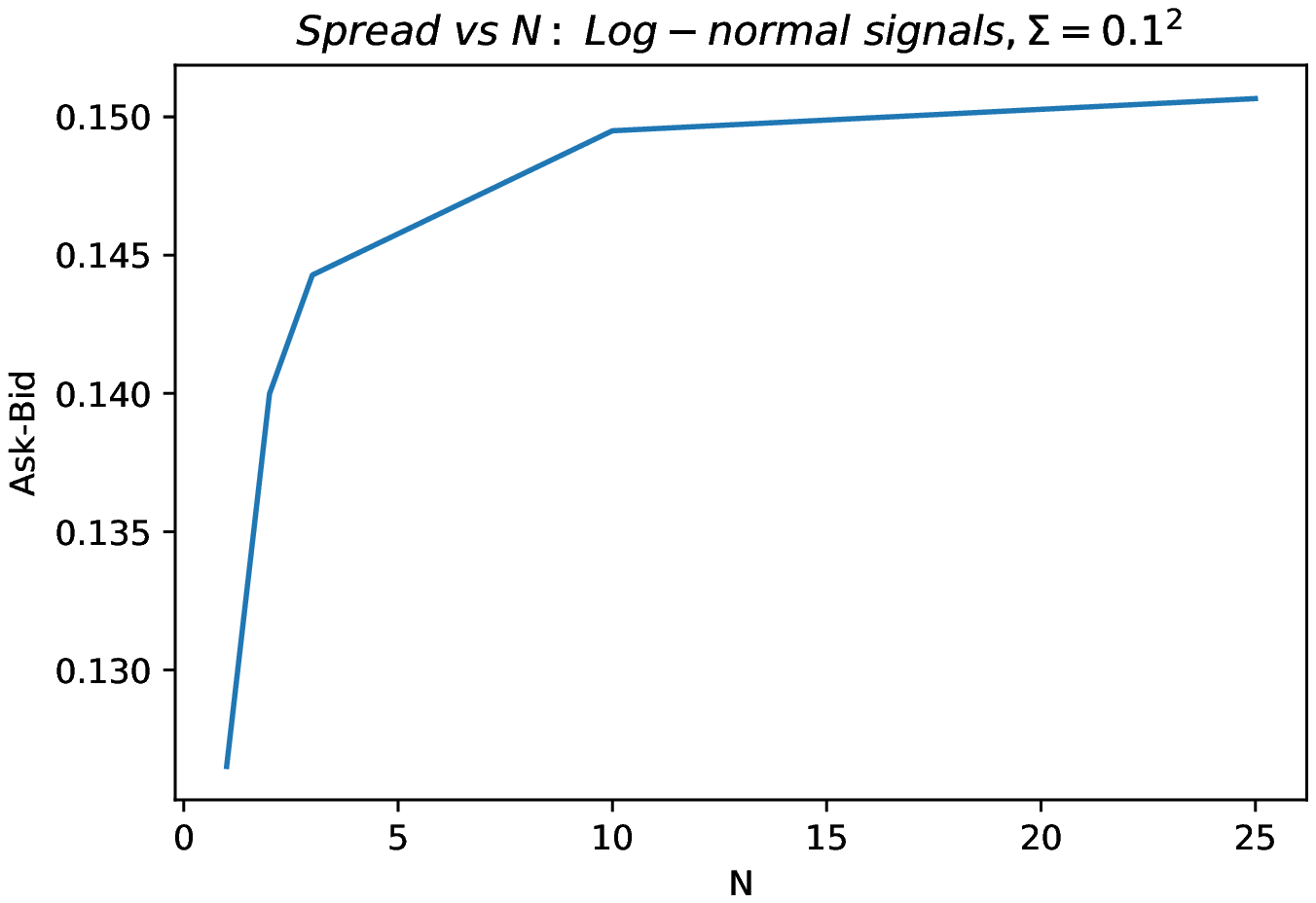}
\caption{The spread is shown as a function of the number of informed investors, for log-normal signals with $\sqrt{\Sigma}=10\%$.}
\centering
\label{fig:LognormalSpreadVsN}
\end{figure}

The greater positive tail mass for the log-normal signals implies that large positive signals yield a greater profit than for Gaussian signals with the same $\Sigma$, and vice-versa, negative signals yield a smaller profit in the log-normal case (Figure ~\ref{fig:LognormalVsGaussianProfit}).

\begin{figure}[h]
	\includegraphics[scale=0.7, angle=0]{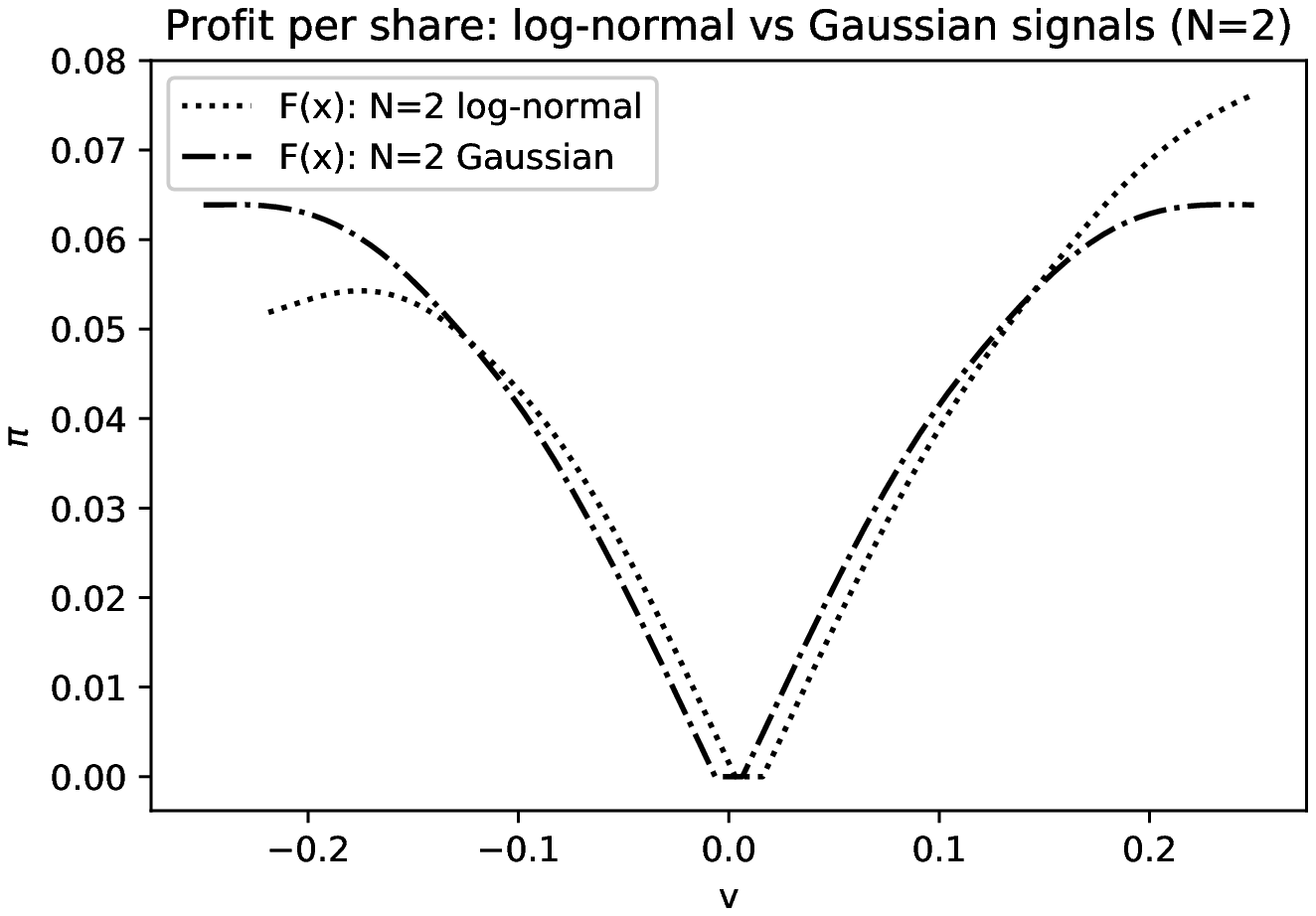}
	\caption{Aggregate profit per share for Gaussian signals vs. Log-normal, for $N=2$}
	\centering
	\label{fig:LognormalVsGaussianProfit}
\end{figure}

\subsubsection{Student signals}

We explore the effect of fat tails in the signal distribution next. We consider the case where signals drawn from a Student t-distribution with $\alpha=3$, $p(v) \sim {{1} \over {(1+{{v^2} \over {\Sigma}})^2}}$. This is reminiscent of some empirical studies such as Plerou \cite{plerou1999}. However, we note that we are assuming a Student distribution of {\it arithmetic} returns. The power-law tail of geometric returns in Plerou's study implies an infinite expected price. 

In view of Table \ref{T:powerlaw}  the expected asymptotics is $F(x)\sim X^{1/(\alpha-1-{\alpha\over N})}$. Moreover, our conjecture predicts no equilibrium when $N\leq \frac{\alpha}{\alpha -1}=\frac{3}{2}$. Indeed, no numerical solution for $F$ can be found  when $N=1$. The equilibrium asymptotically parabolic for $N=2$ (theoretical $\rho^+=2$), linear for $N=3$ (theoretical $\rho^+=1$) and concave for $N \geq 4$. The numerical solutions are shown in Figure ~\ref{fig:StudentSolutions}. For $N=25$, the asymptotic exponent is $F(x)\sim x^{25/47}$ according to our theory. The numerical solution is compared to this prediction in Figure ~\ref{fig:Student25_47}.

\begin{figure}[h]
\includegraphics[scale=0.7, angle=0]{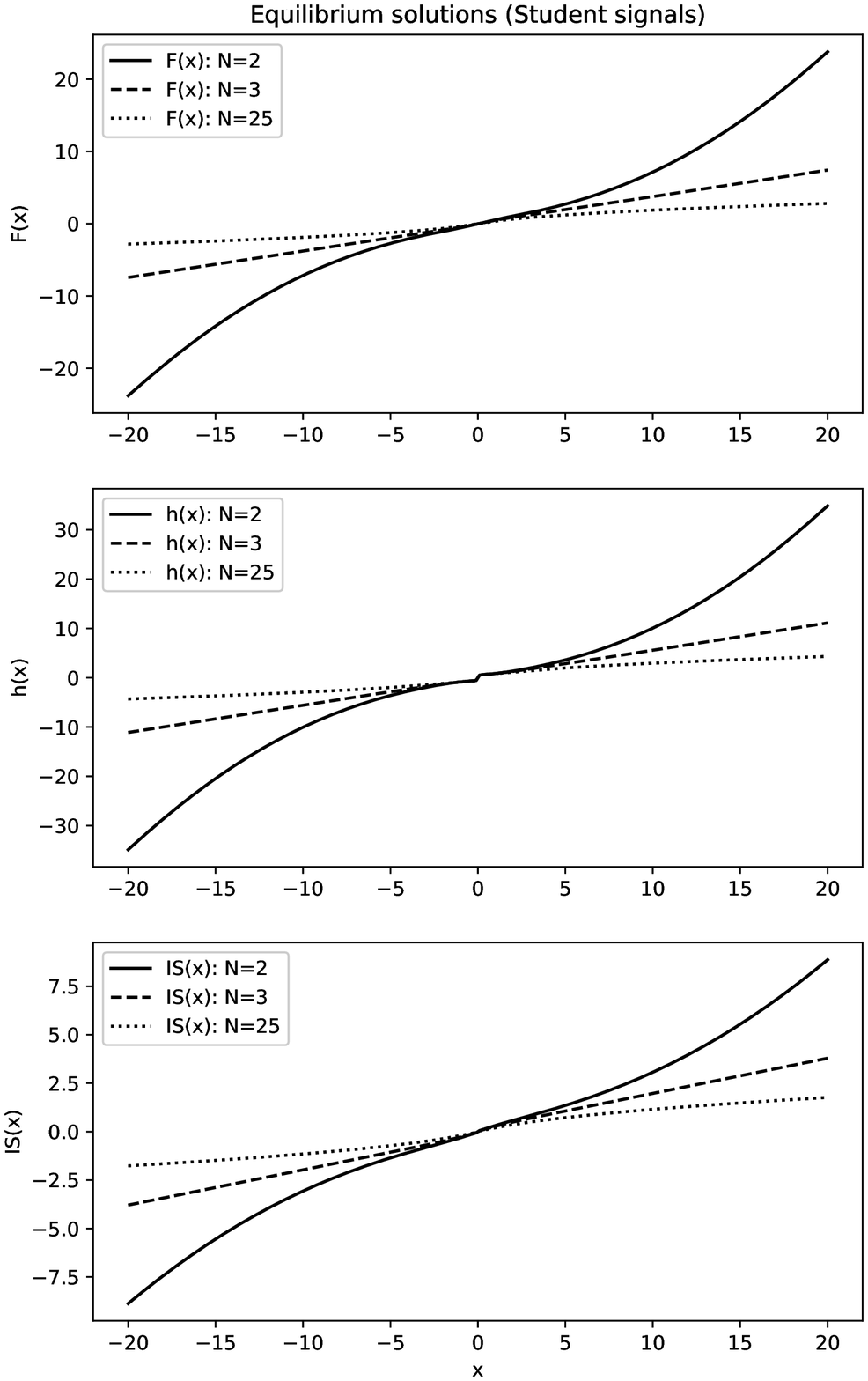}
\caption{Equilibrium solutions for Student signals for the cases $N=2$, $N=3$ and $N=25$}
\centering
\label{fig:StudentSolutions}
\end{figure}

\begin{figure}[h]
\includegraphics[scale=0.7, angle=0]{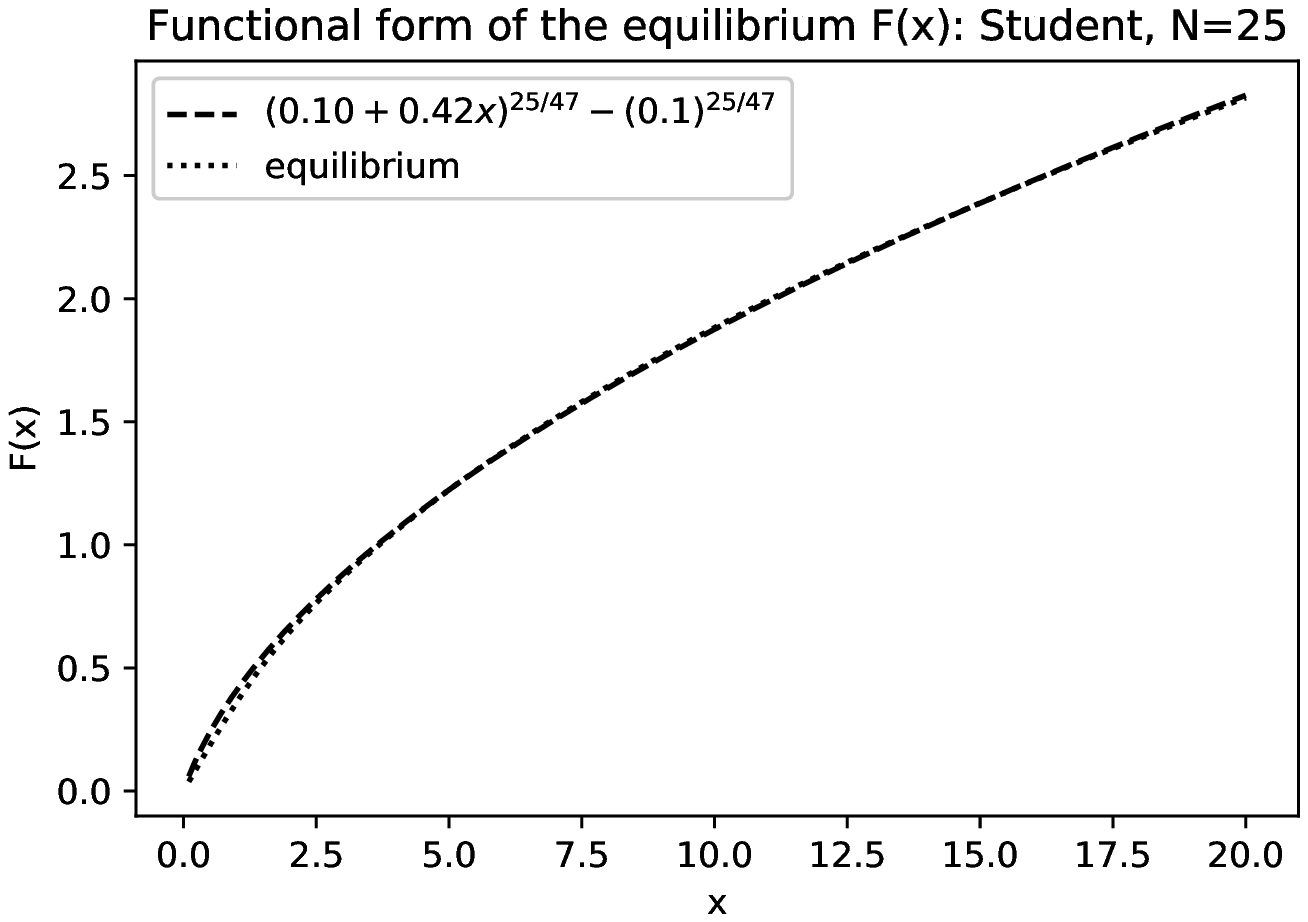}
\caption{Functional form of the equilibrium for Student signals for $\alpha=3$, $N=25$}
\centering
\label{fig:Student25_47}
\end{figure}

The case of power-law tailed signal distributions was considered previously by Farmer et al. in the case of perfect competition between insiders \cite{farmer2013}. One can view their model as the limiting case of the one considered herein as $N\rar \infty$ in case of a Pareto-tailed distribution with exponent $3$. 

\subsubsection{Exponential signals}

We are not aware of any situation that gives rise to an exponential distribution for an asset's price\footnote{A possible exception is Kou and Wang \cite{Kou} who consider an asymmetric double exponential distribution for jumps in the asset price}. However, the case is of interest to illustrate the effect of extreme asymmetry. For exponential signals, the probability density function is given by $p(v)={1\over{\lambda}} e^{-\lambda v}$ for $v>0$. We take $\lambda=1$ and translate the distribution by $1$  so that the mean equals $0$ for the numerical solutions. The numerics indicate that the equilibrium is asymptotically linear for the case of monopolistic insider and concave in case of competitive investors, as shown in Figure ~\ref{fig:ExponentialSolutions}.

\begin{figure}[h]
\includegraphics[scale=0.7, angle=0]{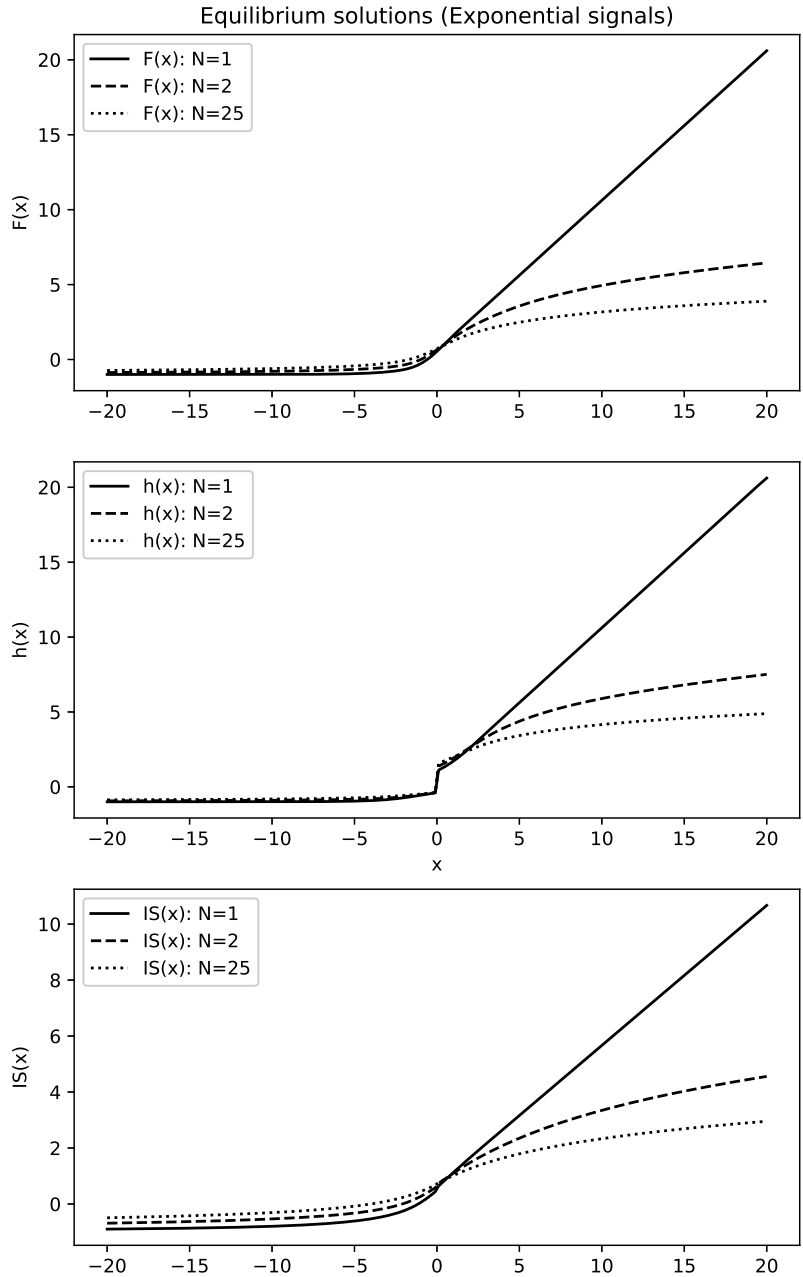}
\caption{Equilibrium solutions for Exponential signals for the cases $N=1$, $N=2$ and $N=25$}
\centering
\label{fig:ExponentialSolutions}
\end{figure}

\section{Same-price liquidation}

Up to now, we assumed that the dealer had an initial position $Z$ and liquidated his total position $X+Z$ after trading with the insider via the limit order book and argued that a Bertrand competition leads to a price to the insider given by (\ref{e:insidercost}). An alternative framework also of interest to practitioners is one where portfolio managers and noise traders submit their orders to an aggregator (institutional trading desk), which merges the orders into a block (``metaorder", in the literature), liquidates $X+Z$ for some average price, and allocates shares with the same average price to noise traders and insiders. We refer to this framework as {\em same-price liquidation}. The expected profit of the insiders in this case is 

\be \label{e:insidercostDemocratic}
E^v \big[ V X- {X \over {X+Z}} \int_0^{X+Z} h(y)dy \big].
\ee

In this case the corresponding first order condition for the maximisation problem is given by $V=F(X^*)$, where 
\[
F(x)=E^v\left[\frac{x}{Z+x}\frac{h(Z+x)-\bar{h}(Z+x)}{N}+\bar{h}(Z+x)\right],
\]
where $\bar{h}(x)=\frac{1}{x}\int_0^x h(y)dy$ and $h$ is given by the tail expectation as above.  

The problem with this first order condition is that it is not clear whether it yields a maximum as $F$ defined above is not necessarily increasing, i.e. we may not  have a concave function to maximise in contrast to the problem studied in previous section.

However, if one assumes that $F$ is increasing and obtains the corresponding integral equation, one can still get a fixed point. Moreover, our numerical solutions always  suggest an increasing solution yielding a `numerical' proof of the existence of equilibrium. 

 The solutions are similar in form and share the same asymptotic behaviour for both bounded and unbounded signal distributions. We show as an example the case of log-normal signals in Figure ~\ref{fig:LognormalSolutionsD}.

\begin{figure}[h]
\includegraphics[scale=0.7, angle=0]{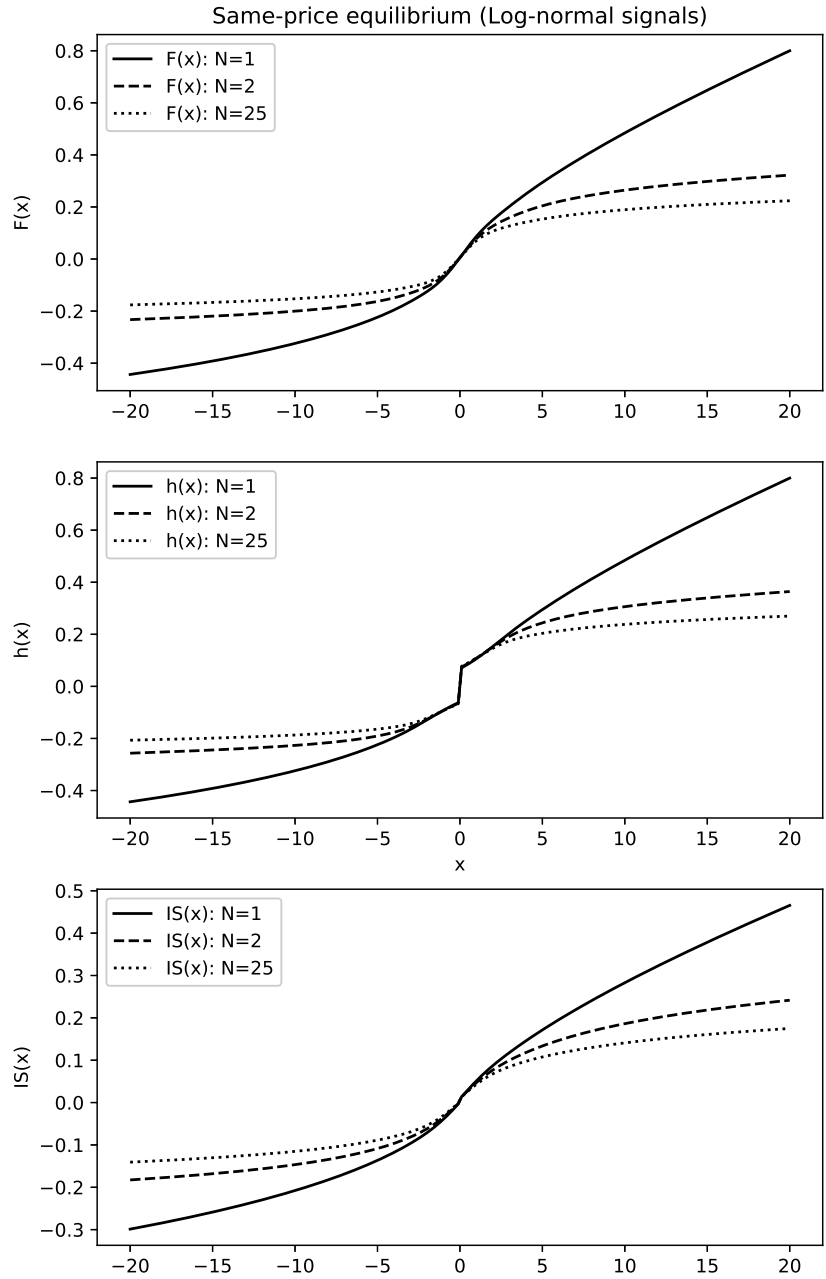}
\caption{Equilibrium solutions for same-price liquidation with log-normal signals for the cases of an insider ($N=1$), and shared signals with $N=2$, $N=25$}
\centering
\label{fig:LognormalSolutionsD}
\end{figure}

Figure ~\ref{fig:SamePriceVsCashDeskGaussian} compares the same-price liquidation model to the dealer inventory model in the case of Gaussian signals. The dotted lines represent the same-price liquidation equilibria. Market impact is somewhat greater with same-price liquidation than for the dealer inventory model, for $N>1$. For the insider case, the two are essentially identical.

\begin{figure}[h]
\includegraphics[scale=0.5, angle=0]{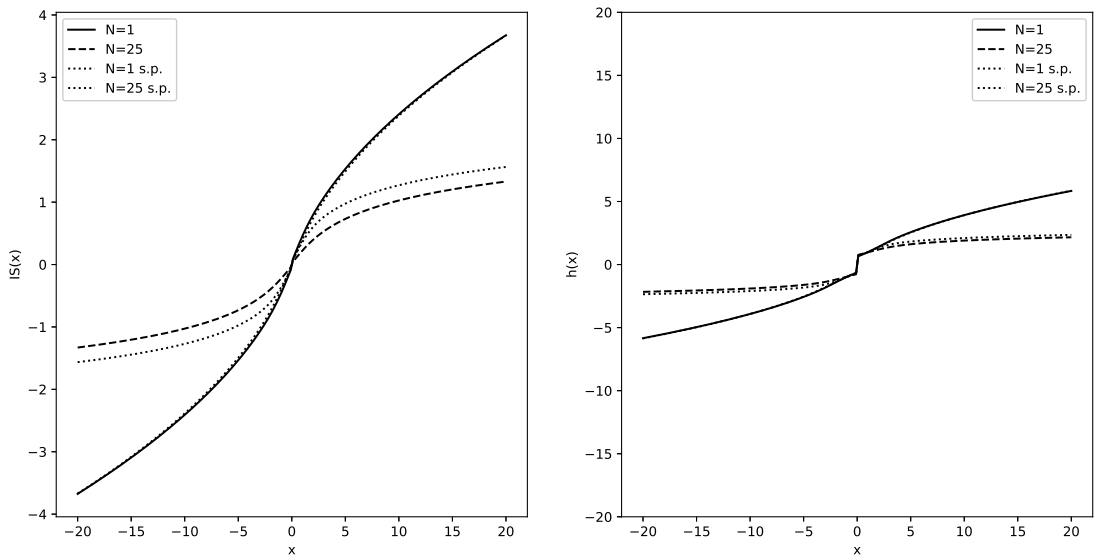}
\caption{Equilibrium for Gaussian signals, comparing the same-price and cash desk liquidation models ($N=1, 2, 25$). Dotted lines represent the same-price liquidation equilibrium}
\centering
\label{fig:SamePriceVsCashDeskGaussian}
\end{figure}

\section{Conclusion}
In this article, we explored how private information is transferred into the market price through a limit order book. Since empirical data on very large trades is sparse and often biased, it is important to develop a theoretical understanding of the process in order to discriminate between various proposals for the shape of the impact function.

We proposed an asymmetric information based equilibrium model where informed investors draw a signal and send their orders to a dealer with an initial position who executes at the net cost to liquidate the aggregate amount against a limit order book.  Unlike the earlier static equilibrium models developed for limit order markets, the informed traders' positions are determined endogenously in equilibrium.  We showed that solutions exist in the case of bounded signals and discussed properties of the equilibrium including the asymptotic behavior of the implementation shortfall for large trades. 

Our results provide the micro-foundations for a  large number of empirical findings  including those on price impact and volume. We found that market impact is asymptotically a power of trade size if the signal has fat tails, whereas the impact becomes of the form $(\log x)^{1/p}$ for some $p>0$ for lighter tails.  

Our numerical experiments show that our results remain  valid for unbounded signals with analogous  impact asymptotics. Moreover, for fat-tailed signal distributions, an equilibrium only exists if there is a sufficient amount of competition.  In the particular case of Student distribution,  if the tail exponent is $\alpha=3$, there is no equilibrium in the case of a monopolistic insider, and the shape of market impact tends to a square root in the limit where the number of informed investors is large, $N\rar \infty$.  Although we do not have an analytic proof, the bid-ask spread seems to be an increasing and bounded function of the number of informed investors. 

A relevant and arguably more realistic extension of our framework while still remaining in a static setting is to consider the scenario in which the insiders receive different but possibly correlated signals regarding the liquidation value. On the other hand, due to our assumption that the insiders' orders arrive simultaneously to the dealer, the optimisation problem of each insider requires the solution of a nonlinear filtering problem even in the case of Gaussian signals. Our present technology is  not yet able to deal with such complications and, therefore, we postpone the discussion of this extension to subsequent research.

In reality limit order markets are dynamic and thus the order books change over time reflecting the changes in market parameters. The analytic characterisation of the equilibrium in the current framework in terms of the fixed point of an integral operator makes us optimistic regarding an extension of the current framework to a dynamic setting in continuous time. However, continuous trading brings extra flexibilities to portfolio choice - including the option to place a market or limit order at each trade - resulting in a more complicated model. This extension, though extremely interesting, will thus be left for future research. 

\bibliographystyle{siam}
\bibliography{ref}

\begin{thebibliography}{10}

\bibitem{Almgren2005}
{\sc R.~Almgren, C.~Thum, E.~Hauptmann, and H.~Li}, {\em Direct estimation of
  equity market impact}, Risk,  (2005).

\bibitem{Back}
{\sc K.~Back}, {\em Insider trading in continuous time}, The Review of
  Financial Studies, 5 (1992), pp.~387--409.

\bibitem{BBLOB}
{\sc K.~Back and S.~Baruch}, {\em Strategic liquidity provision in limit order
  markets}, Econometrica, 81 (2013), pp.~363--392.

\bibitem{Bershova2013}
{\sc N.~Bershova and D.~Rakhlin}, {\em The non-linear market impact of large
  trades: evidence from buy-side order flow}, Quantitative Finance, 13 (2013),
  pp.~1759--1778.

\bibitem{BMR00}
{\sc B.~Biais, D.~Martimort, and J.-C. Rochet}, {\em Competing mechanisms in a
  common value environment}, Econometrica, 68 (2000), pp.~799--837.

\bibitem{BGT}
{\sc N.~H. Bingham, C.~M. Goldie, and J.~L. Teugels}, {\em Regular variation},
  vol.~27, Cambridge university press, 1989.

\bibitem{Dubinsky2006}
{\sc A.~Dubinsky and M.~Johannes}, {\em Fundamental uncertainty, earning
  announcements and equity options}, Columbia University Working paper,
  (2006).

\bibitem{farmer2013}
{\sc J.~D. Farmer, A.~Gerig, F.~Lillo, and H.~Waelbroeck}, {\em How efficiency
  shapes market impact}, Quantitative Finance, 13 (2013), pp.~1743--1758.

\bibitem{Foucault99}
{\sc T.~Foucault}, {\em Order flow composition and trading costs in a dynamic
  limit order market}, Journal of Financial markets, 2 (1999), pp.~99--134.

\bibitem{FM08}
{\sc T.~Foucault and A.~J. Menkveld}, {\em Competition for order flow and smart
  order routing systems}, The Journal of Finance, 63 (2008), pp.~119--158.

\bibitem{Gabaix2006}
{\sc X.~Gabaix, P.~Gopikrishnan, V.~Plerou, and H.~E. Stanley}, {\em
  Institutional investors and stock market volatility}, The Quarterly Journal
  of Economics, 121 (2006), pp.~461--504.

\bibitem{Glosten94}
{\sc L.~R. Glosten}, {\em Is the electronic open limit order book inevitable?},
  The Journal of Finance, 49 (1994), pp.~1127--1161.

\bibitem{powerlawVol}
{\sc P.~Gopikrishnan, V.~Plerou, X.~Gabaix, and H.~E. Stanley}, {\em
  Statistical properties of share volume traded in financial markets}, Physical
  review e, 62 (2000), p.~R4493.

\bibitem{HolSub92}
{\sc C.~W. Holden and A.~Subrahmanyam}, {\em Long-lived private information and
  imperfect competition}, The Journal of Finance, 47 (1992), pp.~247--270.

\bibitem{Kou}
{\sc S.~G. Kou and H.~Wang}, {\em Option pricing under a double exponential
  jump diffusion model}, Management science, 50 (2004), pp.~1178--1192.

\bibitem{Kyle1985}
{\sc A.~S. Kyle}, {\em Continuous auctions and insider trading}, Econometrica,
  53 (1985), pp.~1315--1335.

\bibitem{Lillo2005}
{\sc F.~Lillo, M.~Szabolcs, and J.~D. Farmer}, {\em Theory for long memory in
  supply and demand}, Physical review e, 71 (2005), p.~066122.

\bibitem{ParSep03}
{\sc C.~A. Parlour and D.~J. Seppi}, {\em Liquidity-based competition for order
  flow}, The Review of Financial Studies, 16 (2003), pp.~301--343.

\bibitem{ParSepSurvey}
\leavevmode\vrule height 2pt depth -1.6pt width 23pt, {\em Limit order markets:
  A survey}, Handbook of financial intermediation and banking, 5 (2008),
  pp.~63--95.

\bibitem{peroldIS}
{\sc A.~F. Perold}, {\em The implementation shortfall: Paper versus reality},
  Journal of Portfolio Management, 14 (1988), p.~4.

\bibitem{PBimpact}
{\sc M.~Potters and J.-P. Bouchaud}, {\em More statistical properties of order
  books and price impact}, Physica A: Statistical Mechanics and its
  Applications, 324 (2003), pp.~133--140.

\bibitem{Ranaldo04}
{\sc A.~Ranaldo}, {\em Order aggressiveness in limit order book markets},
  Journal of Financial Markets, 7 (2004), pp.~53--74.

\bibitem{rock}
{\sc K.~Rock}, {\em The specialist's order book and price anomalies}, tech.
  rep., Harvard University, 1989.

\bibitem{Seppi97}
{\sc D.~J. Seppi}, {\em Liquidity provision with limit orders and a strategic
  specialist}, The Review of Financial Studies, 10 (1997), pp.~103--150.

\bibitem{Tarski}
{\sc A.~Tarski}, {\em A lattice-theoretical fixpoint theorem and its
  applications.}, Pacific journal of Mathematics, 5 (1955), pp.~285--309.

\bibitem{Torre97}
{\sc N.~Torre}, {\em BARRA Market impact model handbook}, Berkeley, 1997.

\bibitem{plerou1999}
{\sc P.~V., P.~Gopikrishnan, L.~A. Nunez~Amaral, and E.~H. Stanley}, {\em
  Scaling of the distribution of price fluctuations of individual companies},
  Phys. Rev. E, 60 (1999), p.~6519.

\bibitem{Vaglica2008}
{\sc G.~Vaglica, F.~Lillo, E.~Moro, and R.~N. Mantegna}, {\em Scaling laws of
  strategic behavior and size heterogeneity in agent dynamics}, Physical review
  e, 77 (2008), p.~036110.

\bibitem{Zarinelli15}
{\sc E.~Zarinelli, M.~Treccani, J.~D. Farmer, and F.~Lillo}, {\em Beyond the
  square root: evidence for logarithmic dependence of market impact on size and
  participation rate}, Market Microstructure and Liquidity, 1 (2015).

\end{thebibliography}

\appendix
\section{Auxiliary results}
\begin{lemma} \label{l:Psimonotone}
	$\Psi^+$ and $\Psi^-$ are non-decreasing on the support of $V$.
\end{lemma}
\begin{proof}
	Suppose $x<y$. Note that $\Psi^+$ is non-decreasing if
	\[
	E[V\chf_{[V> y]}]P(V> x)-	E[V\chf_{[V> x]}]P(V> y)\geq 0.
	\]
	Indeed, the left side of the above equals
	\bean
	&&E[(V-y)\chf_{[V> y]}]P(V> x)-E[(V-y)\chf_{[V> x]}]P(V> y)\\
	&=&E[(V-y)\chf_{[V> y]}]\left(P(V> x)-P(V> y)\right)-E[(V-y)\chf_{[x< V\leq y]}]P(V> y),
	\eean
	which is non-negative since $V-y\leq 0$ on the set $[x< V\leq y]$. 
	
	The second assertion is proved analogously.
\end{proof}
\begin{lemma} \label{l:key}
	Let $g:\bbR \to \bbR$ be a  continuous function and $u^{+}$ (resp. $u^-$) be the unique solution of 
	\be \label{e:uRN}
	u_t + \sigma^2 u_{xx}=0, \qquad u(1,x)= \Pi^+(g(z)) \; (\mbox{resp. } u(1,x)= \Pi^-(g(z))).
	\ee
	Then, the following hold:
	\begin{enumerate}
		\item There exits a solution $B$ on a filtered probability space $(\Om, \cF, (\cF_t), \bbQ)$ to the following SDE:
		\be \label{e:Pitrans}
		dB_t=\sigma dW_t +\sigma^2 \frac{u_x(t,B_t)}{u(t,B_t)}dt, \quad B_0=x,
		\ee
		where $u$ is either $u^+$ or $u^-$ and $W$ is a Brownian motion with $W_0=0$.
		\item $\phi^+_g(x)=\bbE^{\bbQ^+}\left[\Psi^+(g(B_1))\right]$ and $\phi^-_ g(x)=\bbE^{\bbQ^-}\left[\Psi^-(g(B_1))\right]$, where $(B,\bbQ^+)$ (resp. $(B,\bbQ^-)$) corresponds to the solution of (\ref{e:Pitrans}) if $u=u^+$ (resp $u=u^-$) and $\bbE^{\bbQ}$ stands for the expectation under $\bbQ$. 
		\item $\phi^+_g(0)> \phi^-_g(0)$.
		\item Suppose further that  $g$ is non-decreasing. Then,  $\phi^{\pm}_g$ are non-decreasing, too. Consequently, $\phi_g$ is non-decreasing. Moreover,
		\bea
		\phi^+_g(x)&\leq& \bbE^{\bbQ^+}\left[\Psi^+(g(\sigma W_1+x))\right]\label{e:phi+bd}\\
		\phi^-_g(x)&\geq& \bbE^{\bbQ^-}\left[\Psi^-(g(\sigma W_1+x))\right] \label{e:phi-bd}.
		\eea
		
	\end{enumerate} 
\end{lemma}

\begin{proof} We shall prove the claims for $u^+$ only, the corresponding proof for $u^-$ being analogous.
	\begin{enumerate}
		\item Note that 
		\[
		u^+(t,x)=\int_{-\infty}^{\infty}\Pi^+(g(z))\frac{1}{\sqrt{2\pi\sigma^2 (1-t)}}\exp\left(-\frac{(y-z)^2}{2\sigma^2(1-t)}\right)dz.
		\]
		Then, if $\beta $ is a  Brownian motion on a filtered probability space $(\Om, \cF, (\cF_t), \bbP)$ with $\beta_0=0$, $u(t,B_t)$ is a bounded martingale with $u(1, B_1)= \Pi^+(g(B_1)$, where $B=\sigma \beta +x$.  Thus, we can define a new measure $\bbQ$ on $(\Om, \cF)$ by 
		\[
		\frac{d\bbQ}{d\bbP}=\frac{u(1, B_1)}{u(0,B_0)}.
		\]
		By means of Girsanov's theorem, under $\bbQ$, $B$ solves (\ref{e:Pitrans}). 
		\item Observe that 
		\bean
		\phi^+ g(x)&=&\frac{\bbE\left[\Psi^+(g(B_1))\Pi^+(g(B_1))\right]}{\bbE\left[\Pi^+(g(B_1))\right]}=\frac{\bbE\left[\Psi^+(g(B_1))u^+(1,B_1)\right]}{\bbE\left[u(1,B_1)\right]}\\
		&=&\bbE^{\bbQ}\left[\Psi^+(g(B_1))\right].
		\eean
		\item The claim is equivalent to
		\[
		\int_{-\infty}^{\infty}\Phi^+(g(z))q(\sigma,z)dz\int_{-\infty}^{\infty}\Pi^-(g(z))q(\sigma,z)dz-\int_{-\infty}^{\infty}\Phi^-(g(z))q(\sigma,z)dz\int_{-\infty}^{\infty}\Pi^+(g(z))q(\sigma,z)dz> 0.
		\]
		Using $\Phi^{+}=E[V]-\Phi^-$ and $\Pi^+ + \Pi^-=1$, the above is valid if and only if
		\bean
		0 &<&\int_{-\infty}^{\infty}\Phi^+(g(z))q(\sigma,z)dz-E[V]\int_{-\infty}^{\infty}\Pi^+(g(z))q(\sigma,z)dz\\
		&=&\int_{-\infty}^{\infty}\left(\Psi^+(g(z))-E[V]\right)q(\sigma,z)\Pi^+(g(z))dz,
		\eean
		which holds since for any $x$ we have $\Psi^+(x)\geq \Psi^+(m)=E[V]$ in view of Lemma \ref{l:Psimonotone} and $\Psi^+$ is not constant.
		\item Now, suppose $g$ is non-decreasing, which in turn implies $u^+_x \leq 0$ since $\Pi^+$ is non-increasing. 
		Therefore, as $\frac{u_x(t,x)}{u(t,x)}$ is Lipschitz on $[0,t]$ for any $t<T$, the standard comparison results for SDEs applied to (\ref{e:Pitrans}) in conjunction with Lemma \ref{l:Psimonotone} imply
		\[
		\bbE^{\bbQ}\left[\Psi^+(g(B_1))\big| B_0=y\right]\geq  \bbE^{\bbQ}\left[\Psi^+(g(B_1))\big| B_0=x\right]\quad \mbox{ if } y\geq x
		\]
		since we can construct all these solutions indexed by their starting point on the same probability space due to the local Lipschitz property of $u_x/u$. This shows the desired monotonicity $\phi^+(g)$. 
		
		Similarly, the same comparison principle yields that the solution of (\ref{e:Pitrans}) is bounded from above by $\sigma W_t +x$ in case of $u=u^+$ since $u^+_x\leq 0$. Combined with the monotonicity property of $\Psi^+(g)$, we deduce $\phi^+_g(x)\leq \bbE^{\bbQ}\left[\Psi^+(g(\sigma W_1+x))\right]$.
	\end{enumerate}
\end{proof}
\section{Proofs}
\begin{proof}[Proof of Proposition \ref{p:wealth}]
\begin{enumerate}
		\item Let $g(x):=E^v[h(x+Z)]$. By direct differentiation, the expression (\ref{e:foc}) implies
		\be\label{e:averageh}
		xF(x)=x\frac{g(x)}{N}+\frac{N-1}{N}\int_0^xg(y)dy,
		\ee which is equivalent to
		\[
		x\frac{g'(x)}{N}+g= F(x)+xF'(x).
		\]
		Recall that $g(0)=E^v[h(Z)]=F(0)$ by construction. Thus, the unique solution of the above ODE with this initial condition is given by
		\[
		g(x)=NF(x)- \frac{N(N-1)}{x^N}\int_0^{x}F(y)y^{N-1}dy=F(x)+ \frac{N(N-1)}{x^N}\int_0^{x}(F(x)-F(y))y^{N-1}dy.
		\] 
		This yields (\ref{e:exph}) after a change of variable.
		\item The above also yields (\ref{e:convergence}) due to  the first order condition $F(X^*)=V$. The remaining assertions  are direct consequences of the strict monotonicity of $F$.
		\item 	Finally, note that the total expected profit is given by
		\bean
		\int_0^{X^*}(v-E^v[h(y+Z)])dy&=&\int_0^{X^*}(v-g(y))dy\\
		&=&vX^*-\frac{N}{N-1}\left(X^*F(X^*)-X^*\frac{g(X^*)}{N}\right)\\
		&=&-\frac{vX^*}{N-1}+\frac{X^*g(X^*)}{N-1}=\frac{X^*}{N-1}\left(g(X^*)-v\right)\\
		&=&	N\int_0^{X^*}(v-F(y))\left(\frac{y}{X^*}\right)^{N-1}dy,
		\eean
		where the second equality follows from (\ref{e:averageh}) and the third is due to $F(X^*)=V$.
	\end{enumerate}
\end{proof}
\begin{proof}[Proof of Proposition \ref{p:IS}]
	Note that $E[h(y+Z)]=E^v[h(y+Z)]$ for all $y$. We shall show the result for $N>1$, the remaining case is similar and easier.
	
	 Using the first representation in (\ref{e:exph}), we obtain
	\bean
	\int_0^x E[h(y+Z)]dy&=& N \int_0^x F(y)dy -N(N-1)\int_0^x dy y^{-N}\int_0^y dz F(z)z^{N-1}\\
	&=& N \int_0^x F(y)dy -N(N-1)\int_0^x dz F(z)z^{N-1} \int_z^x dy y^{-N}\\
	&=& N \int_0^x F(y)dy +N\int_0^x dz F(z)z^{N-1}\left(x^{-N+1}-z^{-N+1}\right)\\
	&=&N\int_0^x dz F(z)\left(\frac{z}{x}\right)^{N-1}=xN \int_0^1 dz F(xz)z^{N-1},
	\eean
	which yields the first assertion once divided by $x$.
	
	The remaining claims follow from $F(xy)<F(x)$ (resp. $F(xy)>F(x)$) for $x>0$ (resp. $x<0$) and for all $y\in (0,1)$ since $F$ is strictly increasing.
\end{proof}
\begin{proof}[Proof of Proposition \ref{p:limitprofit}]
	First observe that $\phi_F$ is bounded since $V$ is a bounded random variable by assumption. Thus the dominated convergence theorem in conjunction with the continuity of $\Pi^+$, and thus that of $\Phi^+$ and $\Phi^-$, show that both of  $\liminf F_N$ and $\limsup F_N$ solve (\ref{e:Flimit}). As (\ref{e:Flimit}) can have at most one solution, $F_{\infty}:=\lim F_N$ exists.
	
	Also note that $F_{\infty}$ is strictly increasing and continuous by Lemma \ref{l:Flim}. Thus, $\lim_{N\rar \infty}F_N^{-1} (v)=F^{-1}_{\infty}(v) \in \bbR.$ Thus,
	\bean
	\lim_{N\rar \infty}\pi^*(v)&=&\lim_{N\rar \infty} F_N^{-1}(v) N\int_0^1 (v- F_N(F_N^{-1}(v) y)y^{N-1}dy\\
	&=& F_{\infty}^{-1}(v)\lim_{N\rar \infty} N\int_0^1 (v- F_N(F_N^{-1}(v) y)y^{N-1}dy
	\eean
	Since each $F_N$ takes values in $(m,M)$, $\lim_{N\rar \infty} F_N(F_N^{-1}(v) y)= F_{\infty}(F_{\infty}^{-1}(v) y)$, and the measure $N y^{N-1}dy$ on $[0,1]$ converges weakly to the point mass at $1$, we have
	\[
	\lim_{N\rar \infty} N\int_0^1 (v- F_N(F_N^{-1}(v) y)y^{N-1}dy= v- F_{\infty}(F_{\infty}^{-1}(v))=0.
	\]
	This completes the proof.
\end{proof}

 \begin{proof}[Proof of Lemma \ref{l:Flim}]
 	We give the proof for the solutions of (\ref{e:F}), the analogous property of the solutions of (\ref{e:Flimit}) can be proven similarly.
 	
 	 Monotone convergence theorem  in conjunction with Assumption \ref{a:Finteg} implies
 	\[
 	\lim_{x\rar\infty }\int_{-\infty}^{\infty}q(\sigma,z)\phi_F(x+z)dz=\int_{-\infty}^{\infty}q(\sigma,z) \lim_{x\rar \infty}\phi_F(x+z)dz.
 	\]
 	Moreover, $\lim_{x\rar \infty}\phi_F(x+z)=\Psi^+(F(\infty)-)$. To see this, first note that 
 	\[
 	\phi^+_F(z+x)=\frac{\int_{-\infty}^{\infty}\Psi^+(F(u))\Pi^+(F(u))\frac{1}{\sqrt{2\pi \sigma^2}}\exp\left(-\frac{(x+z-u)^2}{2\sigma^2}\right)du}{\int_{-\infty}^{\infty}\Pi^+(F(u))\frac{1}{\sqrt{2\pi \sigma^2}}\exp\left(-\frac{(x+z-u)^2}{2\sigma^2}\right)du}.
 	\]
 	Next, the  measure 
 	\[
 	\frac{\Pi^+(F(u)q(\sigma,z+x-u)du}{\int_{-\infty}^{\infty} \Pi^+(F(u)q(\sigma,z+x-u)du}
 	\]
 	converges to the point mass at $\infty$. Indeed, for any $0<a<\infty$, we have
 	\bean
 	\frac{\int_{-\infty}^a \Pi^+(F(u)q(\sigma,z+x-u)du}{\int_{-\infty}^{\infty} \Pi^+(F(u)q(\sigma,z+x-u)du}&=&\frac{\int_{-\infty}^a \Pi^+(F(u)\exp\big(-\frac{u^2}{2\sigma^2}+\frac{u(z+x)}{\sigma^2}\big)du}{\int_{-\infty}^{\infty} \Pi^+(F(u)\exp\big(-\frac{u^2}{2\sigma^2}+\frac{u(z+x)}{\sigma^2}\big)du}\\
 	&	\leq& \frac{\exp(\frac{a(z+x)}{\sigma^2})\int_{-\infty}^a \Pi^+(F(u)\exp(-\frac{u^2}{2\sigma^2})du}{\int_{2a}^{\infty} \Pi^+(F(u)\exp\big(-\frac{u^2}{2\sigma^2}+\frac{u(z+x)}{\sigma^2}\big)du}\\
 	&\leq&\frac{\exp(\frac{-a(z+x)}{\sigma^2})\int_{-\infty}^a \Pi^+(F(u)\exp(-\frac{u^2}{2\sigma^2})du}{\int_{2a}^{\infty} \Pi^+(F(u)\exp\big(-\frac{u^2}{2\sigma^2}\big)du},
 	\eean
 	which converges to $0$ as $x \rar \infty$.
 	
 	Thus, using the representation of $F$ via (\ref{e:FN2}), we deduce
 	\[
 	F(\infty)=\Psi^+(F(\infty)-).
 	\]
 	Note that changing the order of integration is justified thanks to Assumption \ref{a:Finteg}. On the other hand, $\Psi^+(x-)>x$ for any $x<M$ since $P(V>x)>0$ whenever $x<M$. This in turn implies $F(\infty)=M$. Similarly, $\lim_{x\rar -\infty}F(x)=m$. Thus, $F$ is strictly increasing in view of (\ref{e:FN2}) since $m\neq M$ and, therefore, $F$ is not constant. 
 \end{proof}
 
\begin{proof}[Proof of Theorem \ref{t:existence}]
	The proof will be based on an application of Schauder's fixed point theorem to a suitable mapping defined on the space of non-decreasing continuous functions, i.e. the candidate functions for the solution of (\ref{e:F}). Note that since $V$ takes values in $[m,M]$, so does $\Psi^{\pm}$. This justifies the representation (\ref{e:FN2}) for $F$. Since in equilibrium $h$ must also be taking values in $[m,M]$, we must expect $F$ to take values in $[m,M]$, too. Thus, we can concentrate on  functions on $\bbR$ that takes values in $[m,M]$.  Moreover, $F$ will possess a derivative that is bounded by
	\[
	K_0:=(|m|+M) \int_{-\infty}^{\infty}|z|\frac{e^{-\frac{z^2}{2\sigma^2}}}{\sigma^3\sqrt{2\pi}}dz\left(\frac{1}{N}+\frac{N-1}{2N}\right)<\infty.
	\]
	To see this first observe that
	\bean
	\frac{\partial \bar{q}(\sigma,x,z)}{\partial x}&=&\frac{q(\sigma,x-z)-\bar{q}(\sigma,x,z)}{x}\\
	&=&\frac{1}{x^2}\int_0^x\left\{q(\sigma,x-z)-q(\sigma,y-z)dy\right\}\\
	&=&\frac{1}{x^2}\int_0^x u q_x(\sigma,u-z)du.
	\eean
	Therefore, 
	\bean
	\left|\frac{d}{dx}F(x)\right|&\leq& \int_{-\infty}^{\infty}|\phi_F(z)|\left\{\frac{|q_x(\sigma,x-z)|}{N}+\frac{N-1}{Nx^2}\int_0^{|x|}u|q_x(\sigma,u-z)|du \right\}dz\\
	&\leq& (|m|+M)\left\{\frac{1}{N}\int_{-\infty}^{\infty}|z-x|\frac{e^{-\frac{(z-x)^2}{2\sigma^2}}}{\sigma^3\sqrt{2\pi}}dz+ \frac{N-1}{Nx^2}\int_0^{|x|}du u\int_{-\infty}^{\infty}|z-u|\frac{e^{-\frac{(z-u)^2}{2\sigma^2}}}{\sigma^3\sqrt{2\pi}}dz\right\} \\
	&=&(|m|+M) \int_{-\infty}^{\infty}|z|\frac{e^{-\frac{z^2}{2\sigma^2}}}{\sigma^3\sqrt{2\pi}}dz\left(\frac{1}{N}+\frac{N-1}{2N}\right).
	\eean
	We shall show the existence of a fixed point in the normed space $\cX:=L^2(\bbR,\mu_0)$, i.e. the space of Borel measurable functions that are square integrable with respect to $\mu_0$, where
	\[
	\mu_0(dx)=\frac{1}{\sqrt{2\pi}}e^{-\frac{x^2}{2}}dx.
	\]
	Note that $\mu_0$ is equivalent to the Lebesgue measure on $\bbR$. Next define
	\[
	D_0:\{g|g:\bbR \mapsto [m,M] \mbox{ is such that } |g(x)-g(y)|\leq K_0|x-y|, \forall x,y \in \bbR\}
	\]
	and let
	\[
	D:=\{g\in\cX| g=g_0,\, \mu_0\mbox{-a.e. for some } g_0\in D_0\}.\]
	It is easy to see that  $D$ is a convex subset of $\cX$. 
	
	Next define the operator $T$ on $\cX$ via
	\[
	Tg(x):=\int_{-\infty}^{\infty}\left\{\frac{1}{N}q(\sigma,x-z)+\frac{N-1}{N} \bar{q}(\sigma,x,z)\right\}\phi_{\bar{g}}(z)dz,
	\]
	where $\bar{g}:=(g\vee m)\wedge M$ and $\phi_g$  and $\bar{q}$ are as defined in (\ref{d:phig}) and (\ref{e:qbar}), respectively. Note that for each $x$
	\[
	\left\{\frac{1}{N}q(\sigma,x-z)+\frac{N-1}{N} \bar{q}(\sigma,x,z)\right\}dz
	\]
	is a probability measure on $\bbR$.
	
	\begin{itemize}
		\item[Step 1] {\em ($T$ maps $D$ into itself):} It is easy to verify that $Tg$ is continuous and takes values in $[m, M]$  in view of Lemma \ref{l:key} and that 
		\[
		\int_{-\infty}^{\infty}\left\{\frac{1}{N}q(\sigma,x-z)+\frac{N-1}{N} \bar{q}(\sigma,x,z)\right\}dz=1.
		\]
		Moreover, $T_g$ is differentiable with a derivative bounded by $K_0$ by using the above computations that led to the estimate $K_0$. Thus, $Tg\in D_0\subset D$.
		\item[Step 2] {\em ($D$ is compact):} Let $(g_n)\subset D$. Then there exists $(g^0_n)\subset D_0$ such that $\mu_0$-a.e. we have $g_n =g^0_n$ for each $n\geq 1$.  Then by Arzela-Ascoli Theorem there exists a subsequence that converges uniformly on compacts to a continuous function $g^0$. Without loss of generality let us assume that $g^0_n$ converges to $g^0$. Note that necessarily $|g^0(x)-g^0(y)|\leq K_0|x-y|$ for all $x,y \in \bbR$, i.e $g^0\in D_0$.  Finally, as $g_n$s are uniformly bounded and $\mu_0$ is a probability measure, the dominated convergence theorem yields $g_n \rar g^0$ in $L^2(\bbR,\mu_0)$.
		\item[Step 3] {\em ($T:D\to D$ is continuous):} Suppose  $g_n \rar g$ in $D$ as $n \rar \infty$. In view of the definition of $D$ we may assume without loss of generality that that $g_n$s are continuous since changing $g_n$ on a Lebesgue null set does not alter the value of $Tg_n$. By another application of Arzela-Ascoli theorem there exists a subsequence that converges pointwise to some continuous function, which we may identify with $g$ due to the uniqueness of $L^2$-limits up to a null set. Thus, we may assume $g$ is continuous, too. 
		
		Moreover, the same argument shows that every subsequence of $g_n$ has a further subsequence that converges to $g$ pointwise since continuous functions that agree on Lebesgue null sets should agree at every point. Thus, $g_n\rar g$ pointwise as $n\rar \infty$. 
		
		Next, since $\phi_{g_n}$ is uniformly bounded, the dominated convergence theorem yields
		\[
		\lim_{n \rar \infty}Tg_n(x)= \int_{-\infty}^{\infty} dz\left\{\frac{1}{N}q(\sigma,x-z)+\frac{N-1}{N} \bar{q}(\sigma,x,z)\right\}\lim_{n\rar \infty} \phi_{g_n}(z).
		\]
		On the other hand, Lemma \ref{l:key} and Girsanov theorem imply for $z>0$
		\[
		\phi_{g_n}(z)=\frac{\bbE\left[\Psi^+(g_n(\sigma \beta_1+z)) \Pi^+(g_n(\sigma \beta_1+z))\right]}{u^n(0,z)},
		\]
		where $\bbE$ is the expectation operator for $\bbP$ under which $\beta$ is a standard Brownian motion, and $u^n$ is the function $u^+$ in Lemma \ref{l:key} defined by the terminal condition $\Pi^+(g_n)$.  Since $\Psi^+$ and $\Pi^+$ are continuous except on a Lebesgue null set, the dominated convergence theorem yields
		\[
		\lim_{n \rar \infty}\bbE\left[\Psi^+(g_n(\sigma \beta_1+z)) \Pi^+(g_n(\sigma \beta_1+z))\right]=\bbE\left[\Psi^+(g(\sigma \beta_1+z)) \Pi^+(g(\sigma \beta_1+z))\right].
		\]
		Similarly, 
		\[
		\lim_{n \rar \infty}u^n(0,z)=\lim_{n \rar \infty}\bbE\left[ \Pi^+(g_n(\sigma \beta_1+z))\right]=\bbE\left[\Pi^+(g(\sigma \beta_1+z))\right].
		\]
		Thus, we have shown $\lim_{n \rar \infty}\phi_{g_n}(z)=\phi_{g}(z)$ for $z >0$. Analogous arguments yields the convergence for $z\leq 0$, which in turn establishes the pointwise convergence of $Tg_n$ to $T_g$; i.e. 
		\[
		\lim_{n \rar \infty}Tg_n(x)= Tg(x), \qquad x\in \bbR.
		\]
		This yields the claim by an application of the dominated convergence theorem since $Tg_n$s are uniformly bounded. 
	\end{itemize} 
	Therefore, $T$ admits a fixed point in $D$ by Schauder's fixed point theorem. That is, there exists a $g \in D$ such that $g= Tg$. Hence, there exists a solution to (\ref{e:F}). The claim now follows from Theorem \ref{t:eq}.
\end{proof} 
\begin{proof}[Proof of Theorem \ref{t:comparison}]\begin{enumerate}
	\item Let $\mathcal{D}$ be the space of  nondecreasing functions on $\bbR$ that take values in $[m,M]$ and define an operator $T$ on $\mathcal{D}$ via
	\bean
	T r(x)&=&\int_{0}^{\infty}dz\left\{\frac{1}{N}q(\sigma,x-z)+\frac{N-1}{N} \bar{q}(\sigma,x,z)\right\}\int_{-\infty}^{\infty}\Psi^+(r(y))q(\sigma,z-y)dy\nn \\
	&&+ E[V]\int^{0}_{-\infty}dz\left\{\frac{1}{N}q(\sigma,x-z)+\frac{N-1}{N} \bar{q}(\sigma,x,z)\right\}.
	\eean
	First observe that $T r_1\geq T r_2$ if $r_1\leq r_2$ since $\Psi^+$ is non-decreasing. Thus, setting $r_0\equiv m$ and $r_n=T r_{n-1}$ for $n\geq 1$, we observe that $r_n$ is an increasing sequence of nondecreasing continuous functions taking values in $[m,M]$. Thus, the limit, $r_{\infty}$ exists, is nondecreasing,  and  continuous by Dini's theorem. It also follows from the dominated convergence theorem and that $\Psi^+$ has only countably many discontinuities that $r_{\infty}$ is a solution of (\ref{e:compareup}).
	
	Now let $\mathcal{E}$ denote the set of all solutions of (\ref{e:compareup}) in $\mathcal{D}$ and define
	\[
	R^*(x):=\sup_{r\in \mathcal{E}}r(x).
	\]
	Clearly, $R^*$ takes values in $[m,M]$. Due to the aforementioned monotonicity $TR^* \geq Tr=r$ for all $r\in \mathcal{E}$. Consequently, $TR^* \geq R^*$. Setting $g_0=R^*$ and $g_n=T g_{n-1}$ for $n\geq 1$, we again obtain an increasing sequence which converges to some element $r$ of  $\mathcal{E}$. Moreover, $r\geq R^*$, which in turn implies $r=R^*$ by the construction  of $R^*$.
	
	Note that if $R^*$ is constant, it has to equal $M$ since $R^*(\infty)=M$. However, under this assumption, the right side of (\ref{e:compareup}) equals
	\bean
	&&M \int_{0}^{\infty}dz\left\{\frac{1}{N}q(\sigma,x-z)+\frac{N-1}{N} \bar{q}(\sigma,x,z)\right\}\\
	&&+ E[V]\int^{0}_{-\infty}dz\left\{\frac{1}{N}q(\sigma,x-z)+\frac{N-1}{N} \bar{q}(\sigma,x,z)\right\}<M
	\eean
	 as $E[V]<M$. Thus, $R^*$ cannot be constant.
	
	To show the final assertion let $F$ be a solution of (\ref{e:F}). Since $\Psi^- \leq E[V]$ and that 
	$\phi_F^+(z)\leq \int_{-\infty}^{\infty}\Psi^+(F(y))q(\sigma,z-y)dy$ by (\ref{e:phi+bd}), we deduce that
	\[
	F(x)\leq TF(x).
	\]
	As above, setting $g_0=F$ and $g_n=T g_{n-1}$ for $n\geq 1$, we again obtain an increasing sequence which converges to a member of $\mathcal{E}$, which proves the claim.
	\item Since $\Psi^-\leq E[V]$, the proof follows the similar lines as above and, hence, omitted
\end{enumerate}
\end{proof}
\begin{proof}[Proof of Theorem \ref{t:RV}]
	We shall give a proof of the first statement as the second one can be proven along similar lines. We can assume without loss of generality that $E[V]=0$ since, otherwise, we can replace $V$ by $V-E[V]$ and redefine $\Psi^+$ and $R$ accordingly.
	
	By means of a straightforward change of variable we obtain for $x>0$
	\bea 
	\frac{G(\alpha x)}{G(\alpha)}&&=\frac{1}{N}\int_{0}^{\infty}dzq\Big(\frac{\sigma}{\alpha},x-z\Big)\int_{-\infty}^{\infty}dyq\Big(\frac{\sigma}{\alpha},z-y\Big)\frac{M-\Psi^+(R(\alpha y))}{M-R(\alpha)} \nn \\
	&&+\frac{N-1}{N}\int_0^{\infty} dz \bar{q}\Big(\frac{\sigma}{\alpha},x,z\Big) \int_{-\infty}^{\infty}dyq\Big(\frac{\sigma}{\alpha},y-z\Big)\frac{M-\Psi^+(R(\alpha y))}{M-R(\alpha)}.\label{e:Gratio}\\
	&&+\frac{M}{M-R(\alpha)}\int^{0}_{-\infty}dz\left\{\frac{1}{N}q\Big(\frac{\sigma}{\alpha},x-z\Big)+\frac{N-1}{N} \bar{q}\Big(\frac{\sigma}{\alpha},x,z\Big)\right\}\nn
	\eea
	
	Observe that the measure $q(\frac{\sigma}{\alpha},x-z)dz$ converges to the Dirac measure at point $x>0$ as $\alpha \rar \infty$. 
	\begin{enumerate}
		\item[Step 1:] For $y>0$
		\[
		\frac{(M-\Psi^+(R(\alpha y))}{M-R(\alpha)}=\frac{(M-R(\alpha y))\Psi^+_x(y^*)}{M-R(\alpha)}
		\]
		for some $y^*\geq R(\alpha y)$ by the Mean Value Theorem, where $\Psi^+_x$ stands for the derivative of $\Psi^+$, since $\Psi^+(M)=M$. Thus,
		\be \label{e:RVinlim}
		\lim_{\alpha \rar \infty}\frac{(M-\Psi^+(R(\alpha y))}{M-R(\alpha)}=\frac{(M-R(\alpha y))\Psi^+_x(y^*)}{M-R(\alpha)}=\Psi^+_x(M)\lim_{\alpha \rar \infty}\frac{G(\alpha y)}{G(\alpha)}
		\ee
		as $R(\infty)=M$.
		\item[Step 2:] Let $\gamma_{*}(x):=\liminf_{\alpha \rar \infty}\frac{G(\alpha x)}{G(\alpha)}$. Then,  in view of the first step, Fatou's lemma yields
		\[
			\gamma_*(x)\geq  \frac{\Psi^+_x(M)}{N}\gamma_*(x) + \frac{N-1}{Nx}\Psi^+_x(M)\int_0^x \gamma_*(y)dy.\]
		Thus, if $\gamma^*(x)=0$ for some $x>0$, it must be $0$ for almost all $x>0$. However, $\gamma_*(x)\geq 1$ for all $x\in (0,1)$ since $G$ is decreasing. Thus, $\gamma_*>0$ for all $x>0$. In particular, $\gamma_*$ is bounded away from $0$ on $[0,n]$ for any $n\geq 1$. 
		\item[Step 3:] It follows from Step 2 and Corollary 2.0.5 in \cite{BGT} that 
		\[
		\frac{x^d}{C}\leq \frac{G(\alpha x)}{G(\alpha)}\leq C x^c \mbox{ for} x\geq 1 \mbox{ and } \alpha \geq \alpha_0
		\]
		for some constants $c, d,C$ and $\alpha _0$. Moreover, since for $x<1$ we have 
		\[
		\frac{G(\alpha x)}{G(\alpha)}= \left(\frac{G(\alpha x x^{-1})}{G(\alpha x))}\right)^{-1},
		\]
		and $\alpha x >\alpha_0$ for large enough $\alpha$, we deduce that the mapping $(\alpha,x) \mapsto \frac{G(\alpha x)}{G(\alpha)}$ is bounded when $x$ belongs to bounded intervals in $(0,\infty)$.
		\item[Step 4:] Moreover, 
		\[
		(M-R(\alpha))\alpha\geq \frac{N-1}{N}\int_0^{\alpha}du\int_{0}^{\infty}dz q(\sigma,u-z)\int_{-\infty}^{\infty}dy(M-\Psi^+(R(y)))q(\sigma,z-y),
		\]
		which in turn implies
		
		\be \label{e:R*bd}
		\frac{1}{M-R(\alpha)}\leq K\alpha, \qquad \alpha>1 \mbox{ for some } K<\infty.
		\ee
		Since $\frac{G(\alpha x)}{G(\alpha)}$ is bounded when $x$ belongs to bounded intervals in $(0,\infty)$ by Step 3, we obtain, for any $\eps>0$,
		\be
		\lim_{\alpha\rar \infty}\int_{-\infty}^{\infty}dzq\Big(\frac{\sigma}{\alpha},x-z\Big)\int_{z-\eps}^{z+\eps}dyq\Big(\frac{\sigma}{\alpha},z-y\Big)\frac{M-\Psi^+(R(\alpha y))}{M-R(\alpha)} = \Psi^+_x(M)\gamma(x)\label{e:rvlim1},
		\ee
		where $\gamma(x):=\lim_{\alpha\rar \infty}\frac{G(\alpha x)}{G(\alpha)}$ in view of (\ref{e:RVinlim})  provided that the limit exists. 
		
		Furthermore, in view of (\ref{e:R*bd}) we also have
	\bea
	&&\int_{-\infty}^{\infty}dzq\Big(\frac{\sigma}{\alpha},x-z\Big)\int_{\bbR\backslash (z-\eps,z+\eps)}dyq\Big(\frac{\sigma}{\alpha},z-y\Big)\frac{M-\Psi^+(R(\alpha y))}{M-R(\alpha)}\nn\\
	&\leq&K M \int_{-\infty}^{\infty}dzq\Big(\frac{\sigma}{\alpha},x-z\Big)\int_{\bbR\backslash (z-\eps,z+\eps)}dy\alpha q\Big(\frac{\sigma}{\alpha},z-y\Big)\nn\\
	&\rar &0 \mbox{ as } \alpha \rar \infty.\label{e:Rv0lim1}
	\eea
	\item[Step 5:] Applying the arguments of Step 4 to the second and the third integrals in (\ref{e:Gratio}) now shows that for $x>0$
	\be \label{e:ODEgamma}
	\gamma(x)= \frac{\Psi^+_x(M)}{N}\gamma(x) + \frac{N-1}{Nx}\Psi^+_x(M)\int_0^x \gamma(y)dy.
	\ee
	In particular, $\lim_{\alpha\rar \infty}\frac{G(\alpha x)}{G(\alpha)}$ exists. Using the initial condition that $\gamma(1)=1$, direct manipulations show that
	\[
	\gamma(x)=x^{\rho^+}.
	\]
		
	\end{enumerate}
\end{proof}
\begin{proof}[Proof of Theorem \ref{t:SVlog}] We shall again only give the proof the first statement and assume without loss of generality that $E[V]=0$. 
	
	For any $\alpha>0$ define 
	\[
	r(\alpha,x):=\frac{R(\alpha x)-R(\alpha)}{(M-R(\alpha))^{n+1}}.
	\]
	Straightforward manipulations similar to the ones employed in the proof of Theorem \ref{t:RV} leads to
	\bean
	r(\alpha,x)&=&-\frac{R(\alpha)}{(M-R(\alpha))^{n+1}}\int^{0}_{-\infty}dz\left\{\frac{1}{N}q\Big(\frac{\sigma}{\alpha},x-z\Big)+\frac{N-1}{N} \bar{q}\Big(\frac{\sigma}{\alpha},x,z\Big)\right\}\\
	&+&\int_{0}^{\infty}dz\left\{\frac{1}{N}q\Big(\frac{\sigma}{\alpha},x-z\Big)+\frac{N-1}{N} \bar{q}\Big(\frac{\sigma}{\alpha},x,z\Big)\right\}\int_{1}^{\infty}dyq\Big(\frac{\sigma}{\alpha},z-y\Big)\frac{\Psi^+(R(\alpha y))-R(\alpha)}{(M-R(\alpha))^{n+1}}  \\
	&+&\int_{0}^{\infty}dz\left\{\frac{1}{N}q\Big(\frac{\sigma}{\alpha},x-z\Big)+\frac{N-1}{N} \bar{q}\Big(\frac{\sigma}{\alpha},x,z\Big)\right\}\int_{-\infty}^{1}dyq\Big(\frac{\sigma}{\alpha},z-y\Big)\frac{\Psi^+(R(\alpha y))-R(\alpha)}{(M-R(\alpha))^{n+1}}\\
	&\leq&\int_{0}^{\infty}dz\left\{\frac{1}{N}q\Big(\frac{\sigma}{\alpha},x-z\Big)+\frac{N-1}{N} \bar{q}\Big(\frac{\sigma}{\alpha},x,z\Big)\right\}\int_{1}^{\infty}dyq\Big(\frac{\sigma}{\alpha},z-y\Big)\frac{\Psi^+(R(\alpha y))-R(\alpha)}{(M-R(\alpha))^{n+1}}  \\
	&+&\int_{0}^{\infty}dz\left\{\frac{1}{N}q\Big(\frac{\sigma}{\alpha},x-z\Big)+\frac{N-1}{N} \bar{q}\Big(\frac{\sigma}{\alpha},x,z\Big)\right\}\int_{-\infty}^{1}dyq\Big(\frac{\sigma}{\alpha},z-y\Big)\frac{\Psi^+(R(\alpha ))-R(\alpha)}{(M-R(\alpha))^{n+1}}\\
	&\leq&K+\int_{0}^{\infty}dz\left\{\frac{1}{N}q\Big(\frac{\sigma}{\alpha},x-z\Big)+\frac{N-1}{N} \bar{q}\Big(\frac{\sigma}{\alpha},x,z\Big)\right\}\int_{1}^{\infty}dyq\Big(\frac{\sigma}{\alpha},z-y\Big)\frac{\Psi^+(R(\alpha y))-R(\alpha)}{(M-R(\alpha))^{n+1}} 
	\eean
	for some $K>0$ independent of $\alpha$, where the first inequality follows from that $R\geq E[V]=0$ and $R(\alpha y)\leq R(\alpha)$ for $y\leq 1$ and the second is due to the boundedness of $\frac{\Psi^+(R(\alpha y))-R(\alpha)}{(M-R(\alpha))^{n+1}}$.
	
	Moreover, for $y>1$
	\bean
	\frac{\Psi^+(R(\alpha y))-R(\alpha)}{(M-R(\alpha))^{n+1}}&=&\frac{\Psi^+(R(\alpha y))-R(\alpha y)}{(M-R(\alpha))^{n+1}}+r(\alpha,y)\\
	&\leq&\frac{\Psi^+(R(\alpha y))-R(\alpha y)}{(M-R(\alpha y))^{n+1}}+r(\alpha,y)
	\eean
	due to the monotonicity of $R$. Thus, utilising the boundedness of $\frac{\Psi^+(R(\alpha y))-R(\alpha)}{(M-R(\alpha))^{n+1}}$ once more,  we arrive at
	\be 
	r(\alpha,x)\leq \kappa+ \int_{0}^{\infty}dz\left\{\frac{1}{N}q\Big(\frac{\sigma}{\alpha},x-z\Big)+\frac{N-1}{N} \bar{q}\Big(\frac{\sigma}{\alpha},x,z\Big)\right\}\int_{1}^{\infty}dyq\Big(\frac{\sigma}{\alpha},z-y\Big)r(\alpha,y). \label{e:ralphabd} 
	\ee
	for some $\kappa\in (0,\infty)$ that is independent of $\alpha$.
	\begin{enumerate}
		\item[Step 1:] Let $\kappa$ be as above and consider the operator $T: C([1,\infty), [0,\infty)) \to  C([1,\infty), [0,\infty))$ defined by
		\[
		Tf(x)= \kappa+ \int_{0}^{\infty}dz\left\{\frac{1}{N}q\Big(\frac{\sigma}{\alpha},x-z\Big)+\frac{N-1}{N} \bar{q}\Big(\frac{\sigma}{\alpha},x,z\Big)\right\}\int_{1}^{\infty}dyq\Big(\frac{\sigma}{\alpha},z-y\Big)f(y).
		\]
		Clearly $T$ is increasing, i.e. $Tf \geq Tg$ if $f\geq g$.
		\item[Step 2:] For $z\geq 0$ and $c\in \bbR$,
		\be \label{e:innest1}
		\int_c^{\infty}dy q\Big(\frac{\sigma}{\alpha},z-y\Big)y=\frac{\sigma^2}{\alpha^2}q\Big(\frac{\sigma}{\alpha},z-c\Big)+z\int_{c-z}^{\infty}dy q\Big(\frac{\sigma}{\alpha},y\Big)\leq z+ \frac{\sigma}{\alpha\sqrt{2\pi}}.
		\ee
Thus, if $f(x)=\beta x+\delta$ for some $\beta>0$ and $\delta\geq 0$, we have
\bean
Tf(x)&\leq& \delta+ \kappa + \frac{\sqrt{2}\sigma\beta}{\alpha\sqrt{\pi}}+ \frac{\beta x}{N} +\beta \frac{N-1}{x N}\int_0^x ydy\\
&=&\delta +\kappa + \frac{\sqrt{2}\sigma\beta}{\alpha\sqrt{\pi}}+ \frac{\beta x (N+1)}{2 N}\leq \delta+ \kappa + \beta\left(\frac{\sqrt{2}\sigma}{\alpha\sqrt{\pi}}+\frac{3x}{4}\right),
\eean
where the last inequality is due to the hypothesis that $N\geq 2$.

In particular,  given  $\beta =\gamma \kappa$ for some $\gamma>4$,  $Tf \leq f$ whenever $\alpha\geq \frac{4\sqrt{2}\sigma \gamma}{\sqrt{\pi}(\gamma-4)}$.
\item[Step 3:] Define 
\[
\cD(\alpha):=\{f:[1,\infty)\to [0,\infty): f \mbox{ is continuous and } f(x)\leq 5\kappa x +\frac{1}{(M-R(\alpha))^n} \forall x\geq 1\}.
\]
and observe that the restriction of $r(\alpha,\cdot)$ to $[1,\infty)$ belongs to $\cD(\alpha) $ for all $\alpha>0$.

In view of Step 2 we have that $T \cD(\alpha) \subset \cD(\alpha)$ for large enough $\alpha$. Thus, it admits a fixed point in $  \cD(\alpha)$ for large enough $\alpha$ by Tarski's theorem (Theorem 1 in \cite{Tarski}). In fact, it admits a unique fixed point. Indeed, if $f$ and $g$ are two fixed points of $T$ in $ \cD(\alpha)$, 
\[
(f-g)(x)=\int_{0}^{\infty}dz\left\{\frac{1}{N}q\Big(\frac{\sigma}{\alpha},x-z\Big)+\frac{N-1}{N} \bar{q}\Big(\frac{\sigma}{\alpha},x,z\Big)\right\}\int_{1}^{\infty}dyq\Big(\frac{\sigma}{\alpha},z-y\Big)(f-g)(y).
\]
Thus, 
\bean
\inf_{x \geq 1} (f-g)(x)&\geq& \inf_{x \geq 1} (f-g)(x)\inf_{x \geq 1}\int_{0}^{\infty}dz\left\{\frac{1}{N}q\Big(\frac{\sigma}{\alpha},x-z\Big)+\frac{N-1}{N} \bar{q}\Big(\frac{\sigma}{\alpha},x,z\Big)\right\}\int_{1}^{\infty}dyq\Big(\frac{\sigma}{\alpha},z-y\Big)\\
&=&c\inf_{x \geq 1} (f-g)(x)
\eean
for some $c\in (0,1)$. Since $\inf_{x \geq 1} (f-g)(x)<\infty$, the above implies $\inf_{x \geq 1} (f-g)(x)\leq 0$. Applying the same argument to $g-f$, we deduce $f=g$.

Moreover,  taking $f^*(x)=5\kappa x$ and utilising the previous step we deduce that the unique fixed point is bounded from above by $f^*$.  Theorem 1 in \cite{Tarski} now yields that $r(\alpha, \cdot)\leq f^*$ for large enough $\alpha$ in view of (\ref{e:ralphabd}).
\item[Step 4:] In view of Step 3 we have that $r(\alpha,x)\leq 10\kappa$ for all $x\in [1,2]$ for large enough $\alpha$. Thus, Theorem 3.1.5 in \cite{BGT} yields $r(\alpha,\cdot)$ is bounded in bounded subintervals of $(0,\infty)$ uniformly in $\alpha$.
\item[Step 5:] In view of Step 4 and using arguments similar to the ones used in Steps 4 and 5 of the proof of Theorem \ref{t:RV}, we arrive at
\[
\gamma(x)=\frac{1}{k}+ \frac{\gamma(x)}{N}+ \frac{N-1}{xN}\int_0^x \gamma(y)dy,
\] 
where $\gamma(x):=\lim_{\alpha \rar\infty}r(\alpha,x)$ for $x>0$. The unique solution of the above equation with $\gamma(1)=0$ is given by 
\[
\gamma(x)=\frac{N}{k (N-1)}\log x.
\]
\item[Step 6:] Consider
\[
g(x):=\exp\big((M-R(x))^{-n}\big).
\]
Observe that
\bean
\frac{g(\alpha x)}{g(\alpha)}&=&\exp\left(\frac{(M-R(\alpha))^{n}-(M-R(\alpha x))^{		n}}{(M-R(\alpha))^{n}(M-R(\alpha x))^{n}}\right)\\
&=&\exp\left(\frac{(R(\alpha x)-R(\alpha))\sum_{i=0}^{n-1}(M-R(\alpha)^{n-1-i}(M-R(\alpha x)^i}{(M-R(\alpha))^{n}(M-R(\alpha x))^{n}}\right)\\
				&=&\exp\left(\frac{R(\alpha 		x)-R(\alpha)}{(M-R(\alpha))^{n+1}}\frac{M-R(\alpha)}{M-R(\alpha x)}\sum_{i=0}^{n-1}\left(\frac{M-R(\alpha)}{M-R(\alpha x)}\right)^{n-1-i}\right),
\eean
which converges to $\exp(n \gamma(x))$ as $\alpha \rar \infty$ since $M-R$ is slowly varying at $\infty$. Thus, 
\[
g(x)=x^{\frac{N}{N-1}\frac{n}{k}}s(x), \quad x>0.
\]
where $s$ is a slowly varying function at $\infty$.  That is,
\[
(M-R(x))^{-n}= \frac{N}{N-1}\frac{n}{k}\log x+\log s(x).
\]
Since $\lim_{x \rar \infty}\frac{\log s(x)}{\log x}=0$ (cf. Proposition 1.3.6 (i) and (iii) in \cite{BGT}), we have
\[
M-R(x) \sim \left(\frac{N}{N-1}\frac{n}{k}\right)^{-\frac{1}{n}}(\log x)^{-\frac{1}{n}}. \quad x \rar \infty.
\]
	\end{enumerate}
\end{proof}
\begin{proof}[Proof of Corollary \ref{c:SVlog}] Again we only prove the first statement. The hypothesis implies $\Psi^+(x)-x$ is regularly varying with index $n+1$ at $M$. Thus, $\Psi^+(R)-R$ is regularly varying of index $(n+1)\rho^+$ at $\infty$, which in particular implies
	\be \label{e:psirRV}
	\lim_{x\rar \infty}\frac{\Psi^+(R(x))-R(x)}{\int_x^{\infty}\frac{\Psi^+(R(y))-R(y)}{y}}dy=-\lim_{x\rar \infty}\frac{x(\Psi^+_x(R(x)R'(x)-R'(x))}{\Psi^+(R(x))-R(x)}=-(n+1)\rho^+,
	\ee
	in view of Theorem 1.5.11(ii) in \cite{BGT}.
	
	Moreover, direct manipulations yield
	\[
	\Pi^+(x)=\frac{\int_x^M\Pi^+(y)dy}{\Psi^+(x)-x},
		\]
		which in turn implies
		\[
		-\frac{\Pi^+_x(x)}{\Pi^+(x)}=\frac{\Psi^+_x(x)}{\Psi^+(x)-x}.
		\]
		Therefore,
		\bean
			\lim_{x\rar \infty}\frac{x \Pi^+_x(R(x))R'(x)}{\Pi^+(R(x))}&=&-	\lim_{x\rar \infty}\frac{x R'(x)\Psi^+_x(R(x))}{\Psi^+(R(x))-R(x)}\\
		&=&-(n+1)\rho^+-	\lim_{x\rar \infty}\frac{x R'(x)}{\Psi^+(R(x))-R(x)}\\
			&=&-(n+1)\rho^+-k	\lim_{x\rar \infty}\frac{x R'(x)}{(M-R(x))^{n+1}},
		\eean
		where the second equality follows from (\ref{e:psirRV}).
		
		Note that if $\Psi^+_x(M)<1$, $n=0$ and $k=\frac{1}{1-\Psi^+(M)}$ (see Remark \ref{r:remarkn}). In this case, 
		\[
		\lim_{x\rar \infty}\frac{x R'(x)}{M-R(x)}=-\rho^+
		\]
		by Theorem Theorem 1.5.11(ii ) in \cite{BGT}. Thus,
		\[
		\lim_{x\rar \infty}\frac{x \Pi^+_x(R(x))R'(x)}{\Pi^+(R(x)}=\rho^+(k-1)=-\frac{\Psi^+_x(M)}{1-\frac{\Psi^+_x(M)}{N}}.
		\]
		An application of Exercise 1.11.13 in \cite{BGT} to $1/\Pi^+(R)$ establishes the claim.
		
		Now, suppose $\Psi^+_x(M)=1$ and observe that $n$ is necessarily bigger than $0$ and $\rho^+=0$ in this case.  Recall the function $g$ in Step 6 of the proof of Theorem \ref{t:SVlog} and note that
		\[
		\lim_{x\rar \infty}\frac{xg'(x)}{g(x)}=\frac{Nn}{(N-1)k}
		\]
		by another application of Theorem 1.5.11(i) in \cite{BGT}. Since
		\[
		\frac{xg'(x)}{g(x)}=\frac{x n R'(x)}{(M-R(x))^{n+1}},
		\]
		we arrive at
		\[
			\lim_{x\rar \infty}\frac{x R'(x)}{(M-R(x))^{n+1}}=\frac{N}{k(N-1)}.
		\]
		Therefore,
		\[
			\lim_{x\rar \infty}\frac{x \Pi^+_x(R(x))R'(x)}{\Pi^+(R(x)}=-\frac{N}{N-1}=-\frac{\Psi^+_x(M)}{1-\frac{\Psi^+_x(M)}{N}},
		\]
		and we again  conclude by means of Exercise 1.11.13 in \cite{BGT}.
	\end{proof}
\begin{proof}[Proof of Theorem \ref{t:asympF}]
	We shall give a proof of the first statement as the second one can be proven along similar lines. 
	
	By means of a change of variable employed in earlier proofs  we obtain for $x>0$
	\bea 
	\frac{G(\alpha x)}{G(\alpha)}&&=\frac{1}{N}\int_{0}^{\infty}dzq\Big(\frac{\sigma}{\alpha},x-z\Big)\int_{-\infty}^{\infty}\nu(\alpha,z,dy)dy\frac{M-\Psi^+(F(\alpha y))}{M-F(\alpha)} \nn \\
	&&+\frac{N-1}{N}\int_0^{\infty} dz \bar{q}\Big(\frac{\sigma}{\alpha},x,z\Big) \int_{-\infty}^{\infty}\nu(\alpha,z,dy)dy\frac{M-\Psi^+(F(\alpha y))}{M-F(\alpha)}\label{e:GratioF}\\
	&&+\frac{1}{M-F(\alpha)}\int^{0}_{-\infty}dz\left\{\frac{1}{N}q\Big(\frac{\sigma}{\alpha},x-z\Big)+\frac{N-1}{N} \bar{q}\Big(\frac{\sigma}{\alpha},x,z\Big)\right\}\phi_F(\alpha z),\nn
	\eea
	where
	\be \label{e:measnu}
	\nu(\alpha,z,dy):=\frac{\Pi^+(F(\alpha y)q\Big(\frac{\sigma}{\alpha},y-z\Big))}{\int_{-\infty}^{\infty}du\Pi^+(F(\alpha  u)q\Big(\frac{\sigma}{\alpha},u-z\Big))}dy.
	\ee
	
	We shall first demonstrate that for $z>0$ the measure $ \nu(\alpha,z,dy)$ converges to the Dirac measure at point $z$ as $\alpha \rar \infty$.  Indeed, let $R^*$ be the maximal solution of (\ref{e:compareup}) and observe that  for any $\eps>0$ and $z>0$
	\[
	\nu(\alpha,z,(\bbR\backslash(z-\eps,z+\eps)))=\frac{\int_{\bbR\backslash(z-\eps,z+\eps)}dy\Pi^+(F(\alpha y)q\Big(\frac{\sigma}{\alpha},y-z\Big)}{\int_{-\infty}^{\infty}du\Pi^+(F(\alpha  u)q\Big(\frac{\sigma}{\alpha},u-z\Big)}\leq\frac{\int_{\bbR\backslash(z-\eps,z+\eps)}dyq\Big(\frac{\sigma}{\alpha},y-z\Big)}{\int_{0}^{\infty}du\Pi^+(R(\alpha  u)q\Big(\frac{\sigma}{\alpha},u-z\Big)},
	\]
	where the last inequality is due to the fact that $F\leq R^*$ by Theorem \ref{t:comparison}.
	
	Moreover, since $\Pi^+(R)$ is regularly varying at $\infty$ with some index $r\leq 0$ by Corollary \ref{c:SVlog}, we have $\Pi^+(R(\alpha  u)\geq c(1+u)^{r-\delta}\alpha^{r-\delta}$ for some $c>0$ for all $u\geq 0$ by Proposition 1.3.6(v) in \cite{BGT}. Therefore,
	\be \label{e:nubound}
	\nu(\alpha,z,(\bbR\backslash(z-\eps,z+\eps)))\leq \frac{\alpha ^{\delta-r}\int_{\bbR\backslash(z-\eps,z+\eps)}dyq\Big(\frac{\sigma}{\alpha},y-z\Big)}{c\int_{0}^{\infty}du (1+u)^{r-\delta}q\Big(\frac{\sigma}{\alpha},u-z\Big)},
	\ee
	the right side of which converges to $0$ as $\alpha \rar \infty$. Similary, we can show 
	\[
	\lim_{\alpha \rar \infty}\nu(\alpha,z,(-\infty,z-\eps))=0.
	\]
	\begin{enumerate}
		\item[Step 1:] Let $\gamma_{*}(x):=\liminf_{\alpha \rar \infty}\frac{G(\alpha x)}{G(\alpha)}$.  It follows from the same argument in Step 2 of the proof of Theorem \ref{t:RV} that $\gamma_*$ is bounded away from $0$ on $[0,n]$ for any $n\geq 1$. This in turn implies the mapping $(\alpha,x) \mapsto \frac{G(\alpha x)}{G(\alpha)}$ is bounded when $x$ belongs to bounded intervals in $(0,\infty)$ as in Step 3 of the same proof.
		\item[Step 2:] Moreover, since $F\leq R^*$, (\ref{e:R*bd}) yields
		\[
		\frac{1}{M-R(\alpha)}\leq K\alpha, \qquad \alpha>1 \mbox{ for some } K<\infty.\]
Thus, the arguments of Steps 4 and 5 of the proof of Theorem \ref{t:RV} are still applicable due to (\ref{e:nubound}) and that $\Psi^-_F \leq E[V]$.  Consequently, $\gamma$ still satisfies (\ref{e:ODEgamma}), where $\gamma(x)=\lim_{x\rar \infty}\frac{G(\alpha x)}{G(\alpha)}$. In particular, $\gamma(x)=x^{\rho^+}$.
\item[Step 3:] If $\Psi^+_x(M)$, the proof of Theorem \ref{t:SVlog} can be applied verbatim once we show that $f(\alpha,\cdot)	$ is bounded in $[1,2]$ uniformly in $\alpha$, where 
\[
f(\alpha,x):=\frac{F(\alpha x)-F(\alpha)}{(M-F(\alpha))^{n+1}}.
\]
Since $F(\alpha)$ is eventually larger than $E[V]$ and $\phi_F(z)\leq E[V]$ for $z<0$, we have for large enough $\alpha$
\bean
f(\alpha,x)&=&\int^{0}_{-\infty}dz\left\{\frac{1}{N}q\Big(\frac{\sigma}{\alpha},x-z\Big)+\frac{N-1}{N} \bar{q}\Big(\frac{\sigma}{\alpha},x,z\Big)\right\}\frac{\phi_F(\alpha z)-F(\alpha)}{(M-F(\alpha))^{n+1}}\\
&+&\int_{0}^{\infty}dz\left\{\frac{1}{N}q\Big(\frac{\sigma}{\alpha},x-z\Big)+\frac{N-1}{N} \bar{q}\Big(\frac{\sigma}{\alpha},x,z\Big)\right\}\int_{-\infty}^{\infty}\nu(\alpha,z,dy)\frac{\Psi^+(F(\alpha y))-F(\alpha)}{(M-F(\alpha))^{n+1}}  \\
&\leq&\int_{0}^{\infty}dz\left\{\frac{1}{N}q\Big(\frac{\sigma}{\alpha},x-z\Big)+\frac{N-1}{N}\bar{q}\Big(\frac{\sigma}{\alpha},x,z\Big)\right\}\int_{-\infty}^{\infty}q\Big(\frac{\sigma}{\alpha},z-y\Big)\frac{\Psi^+(F(\alpha y))-F(\alpha)}{(M-F(\alpha))^{n+1}},
\eean
where the last inequality follows from Part (4) of Lemma \ref{l:key} since $\Psi^+(F(\alpha y))$ is increasing in $y$ for positive $\alpha$. Thus, $f(\alpha,\cdot)$ can be shown to be bounded by the same function that bounds $r(\alpha,\cdot)$ introduces in the proof of Theorem \ref{t:SVlog}. Repeating the remaining arguments therein yields the claim. 
\end{enumerate}
\end{proof}
\begin{proof}[Proof of Corollary \ref{c:ISasymptotics}]
	It follows from Theorem \ref{t:eq} that $F^*$ must solve (\ref{e:F}). In particular $F^*$ is regularly varying at $\infty$ of order $\rho^+$. 
	
	In view of Proposition \ref{p:IS} we have
	\[
	\frac{M-IS^*(x)}{M-F^*(x)}=N\int_0^1\frac{M-F^*(xy)}{M-F^*(x)}y^{N-1}dy=\frac{N\int_0^x(M-F^*(y))y^{N-1}dy}{x^N(M-F^*(x))}.
	\]
	Observe that $\lim_{y\rar \infty}(M-F^*(y))y^{-\rho^+ \eps}=\infty$ by Proposition 1.3.6 (v) in \cite{BGT}. Thus, $\int_0^{\infty}(M-F^*(y))y^{N-1}dy=\infty$. This justifies the application of the L'Hospital rule to arrive at
	\bean
	\lim_{x\rar \infty}\frac{M-IS^*(x)}{M-F^*(x)}&=&\lim_{x\rar \infty}\frac{N(M-F^*(x))x^{N-1}}{Nx^{N-1}(M-F^*(x))- x^N F^*_x(x)}\\
	&=&\frac{1}{1-\lim_{x\rar \infty} \frac{x F^*_x(x)}{N (M-F^*(x))} }\\
	&=&\frac{1}{1+\frac{\rho^+}{N}}=\frac{N}{N+\rho^+},
	\eean
	where the third equality follows from Exercise 1.11.13 applied to $1/(M-F)$.
	
	Asymptotic relationship near $-\infty$ is proved the same way.
\end{proof}
\end{document}